\documentclass[12pt]{article}

\usepackage{latexsym,amssymb,amsmath}
\usepackage{amsthm, amstext}
\usepackage{array, amsfonts, mathrsfs}
\usepackage[dvips]{graphicx}

\newcommand{\eref}[1]{(\ref{#1})}

\newtheorem{theorem}{Theorem}
\newtheorem{lemma}{Lemma}
\newtheorem{proposition}{Proposition}

\numberwithin{equation}{section}

\date{}
\begin{document}

\author{M.I.Belishev\thanks {Saint-Petersburg Department of
                 the Steklov Mathematical Institute, Saint-Petersburg State University, Russia;
                 belishev@pdmi.ras.ru. 
                 }\,\,
                 and M.N.Demchenko\thanks{Saint-Petersburg Department of
                 the Steklov Mathematical Institute;
                 demchenko@pdmi.ras.ru. 
                 }
                 }

\title{Elements of noncommutative geometry in inverse problems on manifolds}
\maketitle

\begin{abstract}
We deal with two dynamical systems associated with a Riemannian
manifold with boundary. The first one is a system governed by the
scalar wave equation, the second is governed by the Maxwell
equations. Both of the systems are controlled from the boundary.
The inverse problem is to recover the manifold via the relevant
measurements at the boundary (inverse data).

We show that the inverse data determine a C*-algebras, whose
(topologized) spectra are identical to the manifold. By this, to
recover the manifold is to determine a proper algebra from the
inverse data, find its spectrum, and provide the spectrum with a
Riemannian structure.

The paper develops an algebraic version of the boundary control
method, which is an approach to inverse problems based on their
relations to control theory.
\end{abstract}

\setcounter{equation}{0} 

\section{Introduction}
\subsubsection*{About the paper}
One of the basic theses of noncommutative geometry is that a
topological space can be characterized via an algebra associated
with it \cite{Connes}, \cite{Landi}, \cite{RenVar}. In other
words, a space can be encoded into an algebra. As was recognized
in \cite{BCald} and \cite{BSobolev}, such a coding is quite
relevant and efficient for solving inverse problems on manifolds.
In particular, it enables one to reconstruct a Riemannian manifold
via its dynamical or spectral boundary inverse data.

Namely, it is shown that a Riemannian manifold $\Omega$ can be
identified with the (topologized) spectrum $\widehat{{\mathfrak A}(\Omega)}$
of an appropriate Banach algebra ${\mathfrak A}(\Omega)$, the algebra being
determined by the inverse data up to isometric isomorphism.
Therefore, one can reconstruct $\Omega$ by the scheme:
\begin{itemize}
\item extract an isometric copy $\tilde {\mathfrak A}(\Omega)$ of
${\mathfrak A}(\Omega)$ from the data

\item find its  spectrum $\widehat{\tilde {\mathfrak A}(\Omega)}=: \tilde
\Omega$, which is homeomorphic to $\widehat{{\mathfrak A}(\Omega)}$ by
virtue of $\tilde {\mathfrak A}(\Omega)\overset{\rm isom}= {\mathfrak A}(\Omega)$. Thus, we have
$\tilde \Omega\overset{\rm hom}= \Omega$

\item endow $\tilde \Omega$ with a proper Riemannian structure.
\end{itemize}
As a result, we get a Riemannian manifold $\tilde \Omega$
isometric to the original $\Omega$ by construction. It is $\tilde
\Omega$, which solves the reconstruction problem.

Our paper keeps this scheme and extends it to the inverse problem
of electrodynamics.

\subsubsection*{Content}
We deal with a smooth compact Riemannian manifold $\Omega$ with
boundary.
\smallskip

\noindent{\bf Eikonals.}\,\,\,We introduce the {\it eikonals},
which play the role of main instrument for reconstruction. An
eikonal $\tau_\sigma(\cdot)={\rm dist\,}(\cdot, \sigma)$ is a
distance function on $\Omega$ with the base $\sigma \subset
\partial \Omega$. The eikonals determine the Riemannian structure
on $\Omega$.

With each eikonal one associates a self-adjoint operator $\check
\tau_\sigma$ in $L_2(\Omega)$, which multiplies functions by
$\tau_\sigma$. Its representation via the Spectral Theorem is
$\check \tau_\sigma=\int_0^\infty s\,dX^s_\sigma$, where
$X^s_\sigma$ is
the projection onto the subspace $L_2(\Omega^s[\sigma])$ 
of functions supported in the metric neighborhood
$\Omega^s[\sigma] \subset \Omega$ of $\sigma$ of radius $s$.

For an oriented 3d-manifold $\Omega$, by analogy with the scalar
case, we introduce the {\it solenoidal eikonals}
$\varepsilon_\sigma=\int_0^\infty s\,dY^s_\sigma$, which act in the space
${\cal C}=\{{\rm curl\,} h\,|\,\,h, {\rm curl\,} h \in \vec L_2(\Omega)\}$ relevant
to electrodynamics. Here $Y^s_\sigma$ projects vector-fields onto
the subspace of curls supported in $\Omega^s[\sigma]$.
\medskip

\noindent{\bf Algebras.}\,\,\,Eikonals $\{\tau_\sigma\,|\,\,\sigma
\subset \partial \Omega\}$ generate the Banach algebra $C(\Omega)$
of real continuous functions. By the Gelfand theorem, its Gelfand
spectrum (the set of characters) $\widehat{C(\Omega)}$ is
homeomorphic to $\Omega$ \cite{Mur}, \cite{Nai}.

Operator eikonals $\{\check \tau_\sigma\,|\,\,\sigma \subset
\partial \Omega\}$ generate an operator algebra ${\mathfrak T}$, which is a
commutative C*-subalgebra of the bounded operator algebra
${\mathfrak B}(L_2(\Omega))$. The algebras ${\mathfrak T}$ and $C(\Omega)$ are
isometrically isomorphic (via $\check \tau_\sigma \mapsto
\tau_\sigma$). By this, their spectra are homeomorphic, and we
have $\widehat {\mathfrak T} \overset{\rm hom}= \widehat{C(\Omega)} \overset{\rm hom}= \Omega$.

Solenoidal eikonals generate an operator algebra ${\mathfrak E}$, which is a
C*-sub\-algebra of ${\mathfrak B}({\cal C})$. In contrast to ${\mathfrak T}$, the algebra
${\mathfrak E}$ is {\it noncommutative}. However, the factor-algebra $\dot
{\mathfrak E}={\mathfrak E} \slash {\mathfrak K}$ over the ideal of compact operators ${\mathfrak K} \in
{\mathfrak E}$ turns out to be commutative. Moreover, one has $\dot {\mathfrak E}\overset{\rm isom}=
C(\Omega)$ that implies $\widehat {\dot {\mathfrak E}} \overset{\rm hom}= \widehat{C(\Omega)}
\overset{\rm hom}= \Omega$.
\medskip

\noindent{\bf Inverse problems.}\,\,\,Following \cite{BSobolev},
we begin with a dynamical system, which is governed by the scalar
{\it wave equation} in $\Omega$ and controlled from the boundary
$\partial \Omega$. The input$\mapsto$output correspondence is
realized by a {\it response operator} $R$, which plays the role of
inverse data. A reconstruction (inverse) problem is to recover the
manifold $\Omega$ via given $R$.

Solving this problem, we construct (via $R$) an operator algebra
$\tilde {\mathfrak T}$ isometric to ${\mathfrak T}$, find its spectrum $\tilde
\Omega:=\widehat{\tilde {\mathfrak T}} \overset{\rm hom}= \widehat{{\mathfrak T}} \overset{\rm hom}= \Omega$, endow it
with the Riemannian structure by the use of images of eikonals,
and eventually turn $\tilde \Omega$ into an isometric copy of the
original manifold $\Omega$. The copy $\tilde \Omega$ provides the
solution to the reconstruction problem.
\smallskip

In electrodynamics, the corresponding system is governed by the
{\it Maxwell equations} and also controlled from the boundary. The
relevant response operator $R$ plays the role of inverse data for
the reconstruction problem. To solve this problem, we repeat all
the steps of the above described procedure. The only additional
step is the factorization ${\mathfrak E}\mapsto \dot {\mathfrak E}$, which eliminates
noncommutativity.
\medskip

\noindent{\bf Appendix.}\,\,\,Here the basic lemmas on the
eikonals $\varepsilon_\sigma$ and algebra ${\mathfrak E}$ are proven.

\subsubsection*{Comments}
{\bf What is "to recover a manifold"?}\,\,\,Setting the goal to
determine $\Omega$ from $R$, one has to take into account the
evident nonuniqueness of such a determination. Indeed, if two
manifolds $\Omega$ and $\Omega^\prime$ are isometric and have the
mutual boundary $\partial \Omega=\partial \Omega^\prime$ then
their boundary inverse data (in particular, the response
operators) turn out to be identical. Hence, the correspondence
$\Omega \mapsto R$ in not injective and to recover the original
$\Omega$ via $R$ is impossible.

From the physical viewpoint, the inverse data formalize the
measurements, which the external observer implements at the
boundary. The above mentioned nonuniqueness means that the
observer is not able to distinguish $\Omega$ from $\Omega^\prime$
in principle. In such a situation, the only reasonable
understanding of the reconstruction problem is the following: {\it
to construct a manifold $\tilde \Omega$, which possesses the
prescribed inverse data}. It is the above mentioned isometric copy
$\tilde \Omega$, which satisfies this requirement: we have $\tilde
R=R$ by construction.
\medskip

\noindent{\bf Remark}\,\,\,Reconstruction via algebras is known
in Noncommutative Geometry: see \cite{Connes}, \cite{Landi},
\cite{RenVar}. However, there is a principle difference: in the
mentioned papers the starting point for reconstruction is the
so-called {\it spectral triple} $\{{\cal A}, {\cal H}, {\cal D}\}$, which
consists of a commutative algebra, a Hilbert space, and a self-adjoint (Dirac-like)
operator. So, an algebra {\it is given}.

In our case, we at first have to extract an algebra from $R$. Then
we deal with this algebra imposed by inverse data, whereas its
"good" properties are not guaranteed. For instance, a metric graph
is a "commutative space"\, but its eikonal algebra ${\mathfrak T}$ turns out
to be strongly noncommutative \footnote{no factorization turns
${\mathfrak T}$ into a commutative algebra}. The latter leads to
difficulties in reconstruction problem, which are not overcome
yet.
\smallskip

Reconstruction via algebras in inverse problems was originated in
\cite{BCald} and developed in \cite{BSobolev}. It represents an
algebraic version of the {\it boundary control method}, which is
an approach to inverse problems based on their relations to
control theory \cite{BIP97}, \cite{BIP'07}. We hope for further
applications of this version to inverse problems of mathematical
physics.
\medskip

\noindent{\bf Acknowledgements}\,\,\,The authors thank B.A.Plamenevskii for kind and useful consultations.
The work is supported by the
grants RFBR 11-01-00407A, RFBR 12-01-31446, SPbGU 11.38.63.2012, 6.38.670.2013
and RF Government grant 11.G34.31.0026.

\section{Eikonals}
We deal with a real smooth\footnote{everywhere in the paper,
"smooth" means $C^\infty$-smooth} compact Riemannian manifold
$\Omega$ with the boundary $\Gamma$, $g$ is the metric tensor,
${\rm dim\,}\Omega=n\geqslant 2$.

For a set $A \subset \Omega$, by $$\Omega^r[A]:=\{x \in
\Omega\,|\,\,{\rm dist\,}(x,A)<r \},\qquad r>0$$ we denote its metric
$r$-neighborhood. Compactness implies ${\rm diam\,}\Omega:= \\\sup
\{{\rm dist\,}(x,y)\,|\,\,x,y \in \Omega\}<\infty$ and
\begin{equation}\label{Omega r[A]= Omega}\Omega^r[A]=\Omega \qquad {\rm
as}\,\,\,r>{\rm diam\,}\Omega\,.\end{equation}

\subsection{Scalar eikonals}
Let us say a subset $\sigma \subset \Gamma$ to be {\it regular}
and write $\sigma \in {\cal R}(\Gamma)$ if $\sigma$ is diffeomorphic to a
"disk" $\{p \in {\mathbb R}^{n-1}\,|\,\,\|p\|\leqslant 1\}$.

By a (scalar) {\it eikonal} we name a distant function of the form
$$\tau_\sigma (x):={\rm dist\,} (x, \sigma), \qquad x \in \Omega \quad
(\sigma \in {\cal R}(\Gamma))\,.$$ The set $\sigma$ is said to be a {\it
base}. Eikonals are Lipschitz functions: $\tau_\sigma \in {\rm
Lip}(\Omega) \subset C(\Omega)$. Moreover, eikonals are smooth
almost everywhere and \begin{equation}\label{nabla tau=1} |\nabla
\tau_\sigma(x)|=1 \qquad {\rm a.a.}\,\, x \in \Omega
\end{equation} holds. Also, note the following simple geometric facts.
\begin{proposition}\label{Prop 1}
For any $x \in \Omega$ there is $\sigma \in {\cal R}(\Gamma)$ such that
$\tau_\sigma(x)\not=0$. For any different $x, x' \in \Omega$ there
is a $\sigma \in {\cal R}(\Gamma)$ such that
$\tau_\sigma(x)\not=\tau_\sigma(x')$ (i.e., the eikonals
distinguish points of $\Omega$). The equality $\sigma=\{\gamma \in
\Gamma\,|\,\,\tau_\sigma(\gamma)=0\}$ holds.
\end{proposition}

\subsubsection*{Copy $\tilde \Omega$}
As functions on $\Omega$, eikonals are determined by the
Riemannian structure of $\Omega$. The converse is also true in the
following sense.

Assume that we are given with a topological space $\tilde \Omega$,
which is homeomorphic to $\Omega$ (with the Riemann metric
topology) via a homeomorphism $\eta: \Omega \to \tilde \Omega$;
let $\tilde \tau_\sigma:= \tau_\sigma \circ \eta^{-1}$. Also,
assume that $\eta$ {\it is unknown} but we are given with the map
\begin{equation}\label{map sigma to tau sigma} 
{\cal R}(\Gamma) \ni \sigma \mapsto
\tilde \tau_\sigma \in C(\tilde \Omega). \end{equation} Then one can endow $\tilde
\Omega$ with the Riemannian structure, which turns it into a
manifold {\it isometric} to $\Omega$. 
Roughly speaking, the way is the following \footnote{see
\cite{BD_1} for detail}.

For a fixed point $p \in \tilde \Omega$ one can find its neighborhood
$\omega \subset \tilde \Omega$ and the sets $\sigma_1, \dots ,
\sigma_n \in {\cal R}(\Gamma)$ such that the functions $x^1=\tilde
\tau_{\sigma_1}(\,\cdot\,),\, \dots , \,x^n=\tilde
\tau_{\sigma_n}(\,\cdot\,)$ constitute a coordinate chart $\phi:
\omega \ni p \mapsto \{x^k(p)\}_{k=1}^n \in {\mathbb R}^n$. The
coordinates endow $\omega$ with tangent spaces. These spaces can
be provided with the metric tensor $\tilde g=\eta_*g$: one can
determine its components $\tilde g^{ij}$ from the equations
\begin{equation}\label{eqn for g^ik}{\tilde g}^{ij}(x)\,\frac{\partial
\tilde\tau_\sigma \circ \phi^{-1}}{\partial x^i}(x)\,\frac{\partial
\tilde\tau_\sigma \circ \phi^{-1}}{\partial x^j}(x)=1, \qquad x \in
\phi(\omega)\,,\,\,\,\sigma \in {\cal R}(\Gamma)\end{equation} which are just
(\ref{nabla tau=1}) written in coordinates. Choosing here
$\sigma=\sigma_i$, we get $\tilde g^{ii}=1$. Choosing (a finite
number of) additional sets $\sigma$, we can determine the
functions $\frac{\partial \tilde\tau_\sigma \circ \phi^{-1}}{\partial
x^i}$ and then find all other components ${\tilde g}^{ij}(x)$ by
solving the system (\ref{eqn for g^ik}) with respect to them.

So, although the homeomorphism $\eta$ is unknown, we are able to
endow $\tilde \Omega$ with the metric tensor $\tilde g=\eta_* g$,
which turns it into a Riemannian manifold $(\tilde \Omega, \tilde
g)$ isometric to $(\Omega, g)$ by construction.

Moreover, there is a natural way to identify the boundaries
$\tilde \Gamma:=\partial \tilde \Omega$ and $\Gamma=\partial
\Omega$. At first, we can select the boundary points in $\tilde
\Omega$ by $$\tilde \Gamma=\bigcup \limits_{\sigma \in
{\cal R}(\Gamma)}\tilde \sigma, \quad {\rm where}\,\,\,\,\tilde
\sigma:=\{\tilde \gamma \in \tilde \Omega\,|\,\,\tilde
\tau_\sigma(\tilde \gamma)=0\}.$$ Then we identify $\Gamma \ni
\gamma \equiv \tilde \gamma \in \tilde \Gamma$ if $\gamma \in
\sigma$ implies $\tilde \gamma \in \tilde \sigma$ for all regular
$\sigma$ containing $\gamma$.

As a result, we get the manifold $(\tilde \Omega, \tilde g)$
isometric to $(\Omega, g)$, these manifolds having the mutual
boundary $\Gamma$. In what follows we refer to $(\tilde \Omega,
\tilde g)$ as a canonical copy of the original manifold $\Omega$
(shortly: the {\it copy} $\tilde \Omega$).
\medskip

The aforesaid is summarized as follows.
\begin{proposition}\label{Prop 2}
A space $\tilde \Omega$ along with the map 
$(\ref{map sigma to tau sigma})$ determine the copy $\tilde
\Omega$ and, hence, determine $\Omega$ up to isometry of
Riemannian manifolds.
\end{proposition}

\subsection{Operator eikonals}
Introduce the space ${\cal H}:=L_2(\Omega)$ with the inner product
$$(u,v)_{\cal H}\,=\,\int_\Omega u(x) v(x)\, dx\,.$$ Let $A \subset \Omega$ be a
measurable subset, $\chi_A(\,\cdot\,)$ its indicator (a
characteristic function). By
$${\cal H}\langle A \rangle:=\{\chi_A y\,|\,\,y \in {\cal H}\}$$
we denote the subspace of functions supported on $A$. The
(orthogonal) projection $X_A$ in ${\cal H}$ onto ${\cal H}\langle A \rangle$
multiplies functions by $\chi_A$, i.e., cuts off functions on $A$.

Let ${{\mathfrak B}(\cal H)}$ be the normed algebra of bounded operators in ${\cal H}$.
With a scalar eikonal $\tau_\sigma$ one associates an operator
$\check \tau_\sigma \in {{\mathfrak B}(\cal H)}$, which acts in ${\cal H}$ by
$$\left(\check \tau_\sigma y\right)(x)\,:=\,\tau_\sigma(x)\,y(x)\,,
\qquad x \in \Omega\,$$ and is bounded since $\Omega$ is compact.
Moreover, one has \begin{equation}\label{||tau||} \|\check \tau_\sigma\| =
\max_{x \in \Omega}|\tau_\sigma(x)| =\|\tau_\sigma\|_{C(\Omega)}
\leqslant {\rm diam\,}\Omega\,. \end{equation} With a slight abuse of terms,
we also call $\check \tau_\sigma$ an {\it eikonal}.

Each eikonal is a self-adjoint positive operator, which is
represented by the Spectral Theorem in the well-known form.
\begin{proposition}\label{Prop 3}
The representation \begin{equation}\label{spect repr tau} \check
\tau_\sigma\,=\,\int_0^\infty s\,dX^s_\sigma\end{equation} is valid, where
the projections $X^s_\sigma:=X_{\Omega^s[\sigma]}$ cut off
functions on the metric neighborhoods of $\sigma$.
\end{proposition}
Note that the integration interval is in fact $0\leqslant s
\leqslant \|\check \tau_\sigma\|$.

The eikonals corresponding to different bases do commute. This
follows from commutation of $X^s_\sigma$ and $X^{s'}_{\sigma'}$
for all $\sigma, \sigma' \in {\cal R}(\Gamma)$ and $s, s' \geqslant 0$.

\subsection{Solenoidal operator eikonals}
Here we introduce an analog of $\check \tau_\sigma$ relevant to
electrodynamics.

\subsubsection*{3d-manifold}
Now, let ${\rm dim\,}\,\Omega=3$. Also, let $\Omega$ be orientable
and endowed with a volume 3-form $dv$. On such a manifold, the
intrinsic operations of vector analysis
$\wedge$ (vector product), $\nabla,\,{\rm div},\,{\rm curl}$, are well defined on smooth
functions and vector fields (sections of the tangent bundle
$T\Omega$): see, e.g., \cite{Sch}.

\subsubsection*{Solenoidal spaces}
The class of smooth fields $ \vec C^\infty(\Omega)$ is dense in
the space $\vec {\cal H}$ of square-summable fields with the product
$$(a, b)_{\vec {\cal H}}= \int_\Omega a(x)\cdot b(x) \,dx\,,$$ where $\cdot$
is the inner product in $T\Omega_x$. This space contains the
(sub)spaces
$${\cal J}:=\{y \in \vec {\cal H}\,|\,\,{\rm div\,} y=0\,\,{\rm in\,} \Omega\}\,,
\,\,\,\,{\cal C}:=\{{\rm curl\,} h \in \vec {\cal H}\,|\,\,h, {\rm curl\,} h \in \vec
{\cal H}\,\}\,\subset {\cal J}$$ of solenoidal fields and curls. Note that
the smooth classes ${\cal J} \cap \vec C^\infty(\Omega)$ and ${\cal C} \cap
\vec C^\infty(\Omega)$ are dense in ${\cal J}$ and ${\cal C}$ respectively.

Recall the well-known decompositions \begin{equation} \vec {\cal H}= {\cal G}_0\oplus {\cal J}
\,=\,{\cal G}_0 \oplus {\cal C}\oplus {\cal D}\,, \end{equation}\label{Helmholtz} where
${\cal G}_0:=\{\nabla q\,|\,\,q \in H^1_0(\Omega)\}$ is the space of
potential fields, ${\cal D}:=\{y \in {\cal J}\,|\\ \,\,{\rm curl\,} h=0,\,\,\nu
\wedge y=0\,\,{\rm on}\,\,\Gamma\}$ is a finite-dimensional
subspace of harmonic Dirichlet fields \cite{Sch}.

For an $A \subset \Omega$ we denote by \begin{align*} & \vec
{\cal H}\langle A \rangle:=\{\chi_A y\,|\,\,y \in \vec {\cal H}\},\,\,\,{\cal J}\langle A
\rangle:=\overline{\{y \in {\cal J}\,|\,\,{\rm supp\,}y \subset
A\}},\\
& {\cal C}\langle A \rangle:=\overline{\{{\rm curl\,} h\,|\,\,h \in \vec
C^\infty(\Omega),\,\,{\rm supp\,}h \subset A\}}\end{align*} (the
closure in $\vec {\cal H}$) the subspaces of fields supported in $A$.

\subsubsection*{Eikonals $\varepsilon_\sigma$}
Fix a $\sigma \in {\cal R}(\Gamma)$ and take $A=\Omega^s[\sigma]$. Let
$Y^s_\sigma$ be the projection in ${\cal C}$ onto the subspace
${\cal C}\langle\Omega^s[\sigma]\rangle$. Note that the action of $Y^s_\sigma$
is not reduced to cutting off fields on $\Omega^s[\sigma]$, it
acts in more complicated way (see \cite{BIP'07}, \cite{BD_1}).

By analogy with (\ref{spect repr
tau}), define 
a {\it solenoidal operator eikonal} \begin{equation}\label{spec repr eps}
\varepsilon_\sigma\,:=\,\int_0^\infty s\,d Y^s_\sigma\,,\end{equation} which is an
operator in ${\cal C}$. We omit a simple proof of the following result.
\begin{proposition}\label{Prop 4}
The eikonal $\varepsilon_\sigma$ is a bounded self-adjoint positive
operator, the equalities \begin{equation}\label{||eps||}
\|\varepsilon_\sigma\|\,=\,\|\tau_\sigma\|_{C(\Omega)}\,\overset{(\ref{||tau||})}=\,\|\check
\tau_\sigma\| \end{equation} being valid.
\end{proposition}
An important fact is that, in contrast to the cutting off
projections $X^s_\sigma$, the projections $Y^s_\sigma$ and
$Y^{s'}_{\sigma'}$ do not commute in general. As a consequence,
the eikonals $\varepsilon_\sigma$ and $\varepsilon_{\sigma'}$ also {\it do not
commute}.
\smallskip

Multiplying a field $h \in {\cal C}$ by a bounded function $\varphi$,
one takes the field out of the subspace of curls: $\varphi h \in
\vec {\cal H}$ but $\varphi h \not\in {\cal C}$ in general. However, a map
$h \mapsto \varphi h$ is a well defined bounded operator from ${\cal C}$ to $\vec {\cal H}$. 
For instance, understanding  $\check \tau_\sigma$ as an operator,
which multiplies vector fields by the scalar eikonal
$\tau_\sigma$, we have $ \check \tau_\sigma \in {\mathfrak B}({\cal C}; \vec
{\cal H})$.

The following result is of crucial character for future
application to inverse problems. By ${\mathfrak K}({\cal C}; \vec {\cal H}) \subset
{\mathfrak B}({\cal C}; \vec {\cal H})$ we denote the set of compact operators.
\begin{lemma}\label{Lemma 1}
For any $\sigma \subset \Gamma$ the relation $\varepsilon_\sigma - \check
\tau_\sigma \in {\mathfrak K}({\cal C}; \vec {\cal H})$ holds.
\end{lemma}
In the proof (see Appendix) we use the technique developed in~\cite{DMN}.

\section{Algebras}
\subsection{Handbook}
We begin with minimal information about algebras: for detail see,
e.g., \cite{Mur}, \cite{Nai}. The abbreviations BA and CBA mean a
Banach and commutative Banach algebra respectively.
\smallskip

{\bf 1.}\,\,\, A BA is a (complex or real) Banach space ${\cal A}$
equipped with the multiplication operation $ab$ satisfying
$\|ab\|\leq \|a\|\,\|b\|\,\,\, a,b \in {\cal A}$. We deal with
algebras with the unit $e \in {\cal A}:\,\, ea=ae=a$.

A BA ${\cal A}$ is called commutative if $ab=ba$ for all $a,b \in {\cal A}$.
{\it Example}: the algebra $C(X)$ of continuous functions on a
topological space $X$ with the norm
$\|a\|=\sup_{X}|a(\,\cdot\,)|$. The subalgebras of $C(X)$ are
called function algebras.

A CBA is said to be uniform if $\|a^2\|=\|a\|^2$ holds. All
function algebras are uniform.
\smallskip

{\bf 2.}\,\,\, Let ${\cal A}^\prime$ be the space of linear
continuous functionals on a CBA ${\cal A}$. A functional $\delta
\in {\cal A}^\prime$ is called multiplicative if $\delta
(ab)=\delta (a)\delta (b)$. {\it Example}: a Dirac measure
$\delta_{x_0}\in C^\prime (X): \,\delta_{x_0} (a)=a(x_0)
\,\,\,(x_0 \in X)$. Each multiplicative functional is of the norm
1.

The set  of multiplicative functionals endowed with $\ast$-weak
topology (in ${\cal A}^\prime$) is called a {\it spectrum} of
$\cal A$ and denoted by $\widehat{\cal A}$. A spectrum is a compact
Hausdorff space.
\smallskip

{\bf 3.}\,\,\, The {\it Gelfand transform} acts from a CBA $\cal
A$ to $C(\widehat{\cal A})$ by the rule $G\!: a \mapsto
a(\cdot),\,a(\delta):=\delta(a),\, \delta \in \widehat{\cal A}$. It
represents $\cal A$ as a function algebra. The passage from $\cal
A$ to $G{\cal A} \subset C(\widehat{\cal A})$ is referred to as a
geometrization of $\cal A$.
\begin{theorem}{\it(I.M.Gelfand)}
If ${\cal A}$ is a uniform CBA, then $G$ is an isometric
isomorphism from ${\cal A}$ onto $G{{\cal A}}$, i.e., $G(\alpha a
+ \beta b + cd)=\alpha Ga + \beta Gb + Gc\,Gd$ and
$\|Ga\|_{C(\widehat{\cal A})}=\|a\|_{\cal A}$ holds for all $a,b,c,d
\in {{\cal A}}$ and numbers $\alpha, \beta$.
\end{theorem}
\smallskip

{\bf 4.}\,\,\,If two CBA ${\cal A}$ and ${\cal B}$ are
isometrically isomorphic (we write ${\cal A} \overset{\rm isom}=
{\cal B}$) via an isometry $j$, then  the dual isometry $j^\ast:
{{\cal B}}^\prime \to {{\cal A}}^\prime$ provides a homeomorphism
of their spectra:
$j^\ast \widehat{\cal B} = \widehat{\cal A}$. 
Also, one has $G{\cal A} \overset{\rm isom}= G{\cal B}$ via the map $j_\sharp:
Ga \mapsto (Ga)\circ j^*$.

\smallskip

{\bf 5.}\,\,\,Let ${\cal A}(X) \subset C(X)$ be a closed function
algebra. For each $x_0 \in X$, the Dirac measure $\delta_{x_0}$
belongs to $\widehat{{\cal A}(X)}$. Therefore, the map $x_0 \mapsto
\delta_{x_0}$ provides a canonical embedding $X \subset \widehat{{\cal
A} (X)}$.

If $X$ is a compact Hausdorff space, then the Dirac measures
exhaust the spectrum of $C(X)$, whereas the map $x_0 \mapsto
\delta_{x_0}$ provides a canonical homeomorphism from $X$ onto
$\widehat{C(X)}$ (we write $X \overset{\rm hom}= \widehat{C(X)}$). Also, one has
$C(X)\overset{\rm isom}= G C(X)$.
\medskip

The trick, which is used in inverse problems for reconstruction of
manifolds, is the following. Assume that we are given with an
"abstract" CBA ${\mathfrak A}$, which is known to be isometrically
isomorphic to $C(X)$, but neither the (compact Hausdorff) space
$X$ nor the isometry map is given. Then, by determining the
spectrum $\widehat {\mathfrak A}$, we in fact recover the space $X$ up to a
homeomorphism: $X \overset{\rm hom}= \widehat{C(X)}\overset{\rm hom}= \widehat{\mathfrak A}$, whereas
$C(X)\overset{\rm isom}= G C(X)\overset{\rm isom}= G{\mathfrak A}$ does hold. Thus, ${\mathfrak A}$ provides a
homeomorphic copy $\widehat {\mathfrak A}$ of the space $X$ and a concrete
isometric copy $C(\widehat {\mathfrak A})$ of the algebra $C(X)$.
\smallskip

{\bf 6.}\,\,\,A $C^\ast$-algebra is a BA endowed with an {\rm
involution} $(^\ast)$ satisfying $(\alpha a + \beta b + cd)^\ast =
{\bar \alpha}a^\ast + {\bar \beta}b^\ast + d^\ast c^\ast$ and
$\|a^\ast a\|=\|a\|^2$ for all elements $a, b, c, d$ and numbers
$\alpha, \beta$. In the real case, we have just $\bar \alpha =
\alpha$. {\it Example}: the algebra ${\mathfrak B}({\cal H})$ of
bounded operators in a Hilbert space ${\cal H}$ with the operator
norm and conjugation.
\smallskip

{\bf 7.}\,\,\,Let ${\cal I}$ be a norm-closed two-side ideal in a
C*-algebra ${\cal A}$. Then $a\sim b \Leftrightarrow a-b \in {\cal I}$ is an
equivalence. The factor ${\cal A}\slash {\cal I}$ is endowed with a
C*-structure via the projection $\pi: {\cal A} \to {\cal A} \slash {\cal I}$
(element $a$ $\mapsto$ equivalence class of $a$). Namely, one sets
$\|\pi a\|:=\inf \{\|b\|_{\cal A}\,|\,\,b \in \pi a\}, \,\,\alpha \pi
a+\beta \pi b + \pi c \,\pi d :=\pi(\alpha a+\beta b +cd),
\,\,(\pi a)^*:=\pi(a^*)$ for elements $a,b,c,d \in {\cal A}$ and
numbers $\alpha, \beta$. Thus, $\pi$ is a homomorphism of
C*-algebras.

\subsection{Algebra ${{\mathfrak T}}$}
Now let $X$ be our Riemannian manifold $\Omega$, which is
definitely a compact Hausdorff space. Let $C(\Omega)$ be the CBA
of real continuous functions on $\Omega$.

The eikonals $\tau_\sigma$ generate $C(\Omega)$ in the following
sense. For a Banach algebra ${\cal A}$ and a subset $S \subset {\cal A}$, by
$\vee S$ we denote the {\it minimal norm-closed subalgebra of
${\cal A}$, which contains $S$}. The following fact is a
straightforward consequence of the separating properties of
eikonals (Proposition \ref{Prop 1}) and the Stone-Weierstrass
theorem \cite{Nai}.
\begin{proposition}\label{Prop 5}
The equality $\vee\{\tau_\sigma\,|\,\,\sigma \in {\cal R}(\Gamma)\}=C(\Omega)$
is valid.
\end{proposition}

Recall that ${\cal H}=L_2(\Omega)$, ${{\mathfrak B}(\cal H)}$ is the bounded operator
algebra, $\check \tau_\sigma \in {{\mathfrak B}(\cal H)}$ is the multiplication by
$\tau_\sigma$ (see sec 2.2). Introduce the (sub)algebra
\begin{equation}\label{Algebra E (def)}
{\mathfrak T}\,:=\,\vee\{\check \tau_\sigma\,|\,\,\sigma \in
{\cal R}(\Gamma)\}\,\subset\,{{\mathfrak B}(\cal H)} \end{equation} generated by scalar operator
eikonals. As easily follows from (\ref{||tau||}) and Proposition
\ref{Prop 5}, the map $C(\Omega)\ni \tau_\sigma \mapsto \check
\tau_\sigma \in {\mathfrak T}$, which connects the generators, is extended
to an isometric isomorphism of CBA $C(\Omega)$ and ${\mathfrak T}$. With
regard to items {\bf 4, 5} of sec 3.1, the isometry implies
\begin{equation}\label{Omega=Omega E}\Omega \,\overset{\rm hom}=\, \widehat{C(\Omega)}\,\overset{\rm hom}=\,
\widehat{{\mathfrak T}}\,.\end{equation}

\subsubsection*{On reconstruction}
Here we prepare a fragment of the procedure, which will be used
for solving inverse problems.

Assume that we are given with a Hilbert space $\tilde {\cal H}=U{\cal H}$,
where $U$ is a unitary operator. Also assume that we know the map
\begin{equation}\label{map sigma tilde X}{\cal R}(\Gamma) \times [0,T]\,\ni \,\{\sigma,
s\}\,\mapsto \tilde X^s_\sigma \,\in\, {\mathfrak B}(\tilde {\cal H}) \qquad
(T>{\rm diam\,}\Omega)\,,\end{equation} where $\tilde X^s_\sigma := U
X^s_\sigma U^*$, but the operator $U: {\cal H} \to \tilde {\cal H}$ {\it is
unknown} \footnote{in other words, we are given with a
representation of the projection family $\{X^s_\sigma\}_{\sigma
\in {\cal R}(\Gamma)}$ in a space $\tilde {\cal H}$}. Show that this map
determines the manifold $\Omega$ up to isometry. Indeed,

\begin{enumerate}
\item using the map, one can construct the operators $$
\tau^\prime_\sigma:=\int_0^Ts\,d \tilde X^s_\sigma =\int_0^Ts\,d\,
[U X^s_\sigma U^*] \overset{(\ref{spect repr tau})}= U\check
\tau_\sigma U^*$$ \item determine the algebra $\tilde
{\mathfrak T}=\vee\{\tau^\prime_\sigma\,|\,\,\sigma \in {\cal R}(\Gamma)\} \subset
{\mathfrak B}(\tilde{\cal H})\,,$ which is isometric to ${\mathfrak T} \subset {{\mathfrak B}(\cal H)}$ (via
the {\it unknown $U$})

\item applying the Gelfand transform to $\tilde {\mathfrak T}$, find its
spectrum $\widehat{\tilde {\mathfrak T}}=:\tilde \Omega$ and the functions
$\tilde \tau_\sigma:=G \tau^\prime_\sigma$ on $\tilde \Omega$.

\end{enumerate}

Since $\tilde {\mathfrak T} \overset{\rm isom}= {\mathfrak T}$, one has $\tilde \Omega :=\widehat{\tilde
{\mathfrak T}}\overset{\rm hom}= \widehat{{\mathfrak T}}\overset{\rm hom}= \Omega$ (see (\ref{Omega=Omega E})).
Hence, we get a homeomorphic copy $\tilde \Omega$ of the original
$\Omega$ along with the images $\tilde \tau_\sigma$ of the
original eikonals $\tau_\sigma$ on $\Omega$ \footnote{by
construction, $\tilde \tau_\sigma$ turns out to be a pull-back
function of $\tau_\sigma$ via the homeomorphism $\tilde \Omega \to
\Omega$}. Thus, we have a version of the map 
(\ref{map sigma to tau sigma}), which determines the copy $\tilde
\Omega$ (see Proposition \ref{Prop 2}).

Summarizing, we arrive at the following assertion.
\begin{proposition}\label{Prop 6}
The map $(\ref{map sigma tilde X})$ determines the copy $\tilde
\Omega$ and, hence, determines $\Omega$ up to isometry of
Riemannian manifolds.
\end{proposition}
Moreover, the procedure 1.-- 3. provides the copy $\tilde \Omega$.

\subsection{Algebra ${{\mathfrak E}}$}
Recall that the eikonals $\varepsilon_\sigma$ are introduced on a
3d-manifold $\Omega$ by (\ref{spec repr eps}).
\smallskip

An operator (sub)algebra \begin{equation}\label{Algebra E sol
(def)}{\mathfrak E}\,:=\,\vee\{\varepsilon_\sigma\,|\,\,\sigma \in
{\cal R}(\Gamma)\}\,\subset\,{\mathfrak B}({\cal C})\end{equation} is a "solenoidal" analog of the
algebra ${\mathfrak T}$ defined by (\ref{Algebra E (def)}). It is a real
algebra generated by self-adjoint operators. As such, ${\mathfrak E}$ is a
C*-algebra. In contrast to ${\mathfrak T}$, the algebra ${\mathfrak E}$ is not
commutative (see the remark below Proposition \ref{Prop 4}).
However, this non-commutativity is weak in the following sense.

Let ${\mathfrak K} \subset{\mathfrak B}({\cal C})$ be the ideal of compact operators.
Denote ${\mathfrak K}[{\mathfrak E}]:= {\mathfrak K} \cap {{\mathfrak E}}$ and $\dot {\mathfrak E}:={\mathfrak E}\slash
{\mathfrak K}[{\mathfrak E}]$;
let 
$\pi:{\mathfrak B}({\cal C})\to{\mathfrak B}({\cal C}) / {\mathfrak K}$
be the canonical projection.
By (\ref{Algebra E sol (def)}), the latter factor-algebra is
generated by the equivalence classes of eikonals:
$$\dot {\mathfrak E}\,:=\,\vee\{\pi \varepsilon_\sigma\,|\,\,\sigma \in
{\cal R}(\Gamma)\}.$$ Recall that the eikonals $\tau_\sigma$ generate the
algebra $C(\Omega)$: see Proposition \ref{Prop 5}.
\begin{theorem}\label{Theorem 2}
$\dot {\mathfrak E}$ is a commutative C*-algebra. The map
$$C(\Omega) \ni \tau_\sigma \mapsto \pi \varepsilon_\sigma \in \dot {\mathfrak E}
\qquad (\sigma \in {\cal R}(\Gamma)),$$ which relates the generators, can be extended
to an isometric isomorphism from $C(\Omega)$ onto $\dot {\mathfrak E}$.
\end{theorem}
\begin{proof}
Define a map
$$
    \dot\pi : C(\Omega) \to {\mathfrak B}({\cal C}) / {\mathfrak K}
$$
in the following way. Let $Y$ be the projection on ${\cal C}$
acting in $\vec{\cal H}$. With a function $f\in
C(\Omega)$ we associate an operator $Y[f] \in {\mathfrak B}({\cal
C})$ acting by
$$
    Y[f]\, y := Y(f y), \quad y\in \cal C.
$$
Now, define
$$
    \dot\pi(f) := \pi(Y[f]).
$$

For $f\in C(\Omega)$ we denote by $\check f$ the operator
in $\vec{\cal H}$, which multiplies fields by $f$.
The following two Lemmas are proved in Appendix.
\begin{lemma}
    \label{Eff}
    For any $f\in C(\Omega)$ we have
    $$
        \check f - Y[f] \in {\mathfrak K}({\cal C}; \vec{\cal H}).
    $$
\end{lemma}
\begin{lemma}\label{isomorphism}
    The mapping $\dot\pi$
    is an injective homomorphism of C*-algebras.
    \label{hatpihom}
\end{lemma}

To prove Theorem~\ref{Theorem 2} it suffices to
show that the map $\dot\pi$ is an extension of the map
$\tau_\sigma \mapsto \pi \varepsilon_\sigma$. Toward
this end, let us show that $
    \varepsilon_\sigma - Y[\tau_\sigma] \in {\mathfrak K}.
$ Indeed, we have
$$
    \varepsilon_\sigma - Y[\tau_\sigma] = \varepsilon_\sigma - \check\tau_\sigma + \check\tau_\sigma - Y[\tau_\sigma]
$$
and, due to Lemmas~\ref{Lemma 1} and \ref{Eff}, there
is a sum of two compact operators from ${\mathfrak K}(\cal C;
\vec H)$ in the right hand side. Now Theorem~\ref{Theorem 2}
follows from Lemma~\ref{hatpihom} and the fact that algebra $\dot
{\mathfrak E}$ is generated by elements $\pi \varepsilon_\sigma$.
\end{proof}

With regard to items {\bf 4, 5} of sec 3.1, the relation
$C(\Omega)\overset{\rm isom}= \dot {\mathfrak E}$ established by Theorem \ref{Theorem 2}
implies \begin{equation}\label{Omega=Omega dot E sol}\Omega \,\overset{\rm hom}=\,
\widehat{C(\Omega)}\,\overset{\rm hom}=\, \widehat{\dot{\mathfrak E}}\,.\end{equation}
\medskip

\noindent{\bf Remark}\,\,\,Examples, in which
factorization eliminates noncommutativity, are well known. For
instance,
let $X$ be a compact smooth manifold (without boundary) and let
${\mathfrak A} \subset {\mathfrak B}(L_2(X))$ be a C*-algebra generated by
a certain class of pseudo-differential operators of order $0$.
Then the factor-algebra ${\mathfrak A}/{\mathfrak K}$ 
is commutative and isomorphic to the algebra of continuous functions on the
cosphere bundle of $X$ (see \cite{Plamen}).

\subsubsection*{On reconstruction}
Here we provide an analog of the procedure described in sec 3.2.
This analog is relevant to inverse problems of electrodynamics.
Recall that $Y^s_\sigma$ is the projection in ${\cal C}$ onto the
subspace ${\cal C}\langle\Omega^s[\sigma]\rangle$.

Assume that we are given with a Hilbert space $\tilde {\cal C}=U{\cal C}$,
where $U$ is a unitary operator. Also assume that we know the map
\begin{equation}\label{map sigma tilde Y}{\cal R}(\Gamma) \times [0,T]\,\ni \,\{\sigma,
s\}\,\mapsto \tilde Y^s_\sigma \,\in\, {\mathfrak B}(\tilde {\cal C}) \qquad
(T>{\rm diam\,}\Omega)\,,\end{equation} where $\tilde Y^s_\sigma := U
Y^s_\sigma U^*$, but the operator $U: {\cal C} \to \tilde {\cal C}$ is {\it
unknown}. Show that this map determines the manifold $\Omega$ up
to isometry. Indeed,

\begin{enumerate}
\item using the map, one can construct the operators $$
\varepsilon^\prime_\sigma:=\int_0^Ts\,d \tilde Y^s_\sigma =\int_0^Ts\,d\,
[U Y^s_\sigma U^*] \overset{(\ref{spec repr eps})}= U\check
\varepsilon_\sigma U^*$$ \item determine the algebra
${\mathfrak E}^\prime=\vee\{\varepsilon^\prime_\sigma\,|\,\,\sigma \in {\cal R}(\Gamma)\}
\subset {\mathfrak B}(\tilde{\cal C})\,,$ which is isometric to ${\mathfrak E} \subset
{\mathfrak B}({\cal C})$ (via {\it unknown $U$})

\item construct the factor-algebra $\tilde {\mathfrak E}:={\mathfrak E}^\prime \slash
{\mathfrak K}[{\mathfrak E}^\prime]$ over the compact operator ideal in ${\mathfrak E}^\prime$.
By construction, one has $\tilde {\mathfrak E} \overset{\rm isom}= {\mathfrak E}\slash
{\mathfrak K}[{\mathfrak E}]=:\dot {\mathfrak E}$.

\item applying the Gelfand transform to $\tilde {\mathfrak E}$, find its
spectrum $\widehat{\tilde {\mathfrak E}}=:\tilde \Omega$ and the functions
$\tilde \tau_\sigma:=G \pi \varepsilon^\prime_\sigma$ on $\tilde \Omega$.

\end{enumerate}

Since $\tilde {\mathfrak E} \overset{\rm isom}= \dot {\mathfrak E}$, one has $$\tilde \Omega
:=\widehat{\tilde {\mathfrak E}}\overset{\rm hom}= \widehat{\dot {\mathfrak E}}\overset{\rm hom}= \Omega$$ (see
(\ref{Omega=Omega dot E sol})). So, we get a homeomorphic copy
$\tilde \Omega$ of the original $\Omega$ along with the images
$\tilde \tau_\sigma$ of the original
eikonals $\tau_\sigma$ on $\Omega$. Thus, we have a version of the map 
(\ref{map sigma to tau sigma}). This map determines the Riemannian
structure on $\tilde \Omega$, which turns it into an isometric
copy of $\Omega$ (see Proposition \ref{Prop 2}).

Summarizing, we arrive at the following.
\begin{proposition}\label{Prop 7}
The map $(\ref{map sigma tilde Y})$ determines the copy $\tilde
\Omega$ and, hence, determines $\Omega$ up to isometry of
Riemannian manifolds.
\end{proposition}
Moreover, the procedure 1.-- 4. enables one to construct the copy
$\tilde \Omega$. This procedure differs from its scalar analog by
one additional step that is factorization.

\section{Inverse problems}
\subsection{Acoustical system}
With the manifold $\Omega$ one associates a dynamical system
$\alpha^T$ of the form
\begin{align}\label{acoust 1}
 & u_{tt}-\Delta u=0 &&{\rm in}\,\,\,(\Omega\backslash\Gamma) \times
(0,T)\\
\label{acoust 2}
 &u|_{t=0}=u_t|_{t=0}\,=\,0 &&{\rm in}\,\,\,\,\Omega\\
\label{acoust 3} &  u=f  &&{\rm on}\,\,\,\Gamma \times [0,T],
\end{align}
where $\Delta$ is the (scalar) Beltrami--Laplace operator, $t=T>0$
is a final time, $f$ is a {\it boundary control}, $u=u^f(x,t)$ is
a solution. For controls of the smooth class
$${\cal M}^T:=\{f \in C^\infty(\Gamma \times [0,T])\,|\,\,{\rm supp\,}
f \subset \Gamma \times (0,T]\}$$ problem (\ref{acoust
1})--(\ref{acoust 3}) has a unique classical (smooth) solution
$u^f$. Note that the condition on ${\rm supp\,}f$ means that $f$
vanishes near $t=0$.

>From the physical viewpoint, $u^f$ can be interpreted as an
acoustical {\it wave}, which is initiated by the boundary sound
source $f$ and propagates into a domain $\Omega$ filled with an
inhomogeneous medium.

\subsubsection*{Attributes}
$\bullet$\,\,\,The space of controls ${\cal F}^T:= L_2\left(\Gamma
\times [0,T]\right)$ is said to be an {\it outer space} of the
system $\alpha^T$. The smooth class ${\cal M}^T$ is dense in ${\cal F}^T$.

The outer space contains the subspaces
$${\cal F}^{T, s}_\sigma:= \{f \in {\cal F}^T\,|\,\,{\rm supp\,}f \subset
\sigma \times [T-s,T]\}, \qquad \sigma \in {\cal R}(\Gamma).$$ Such
a subspace consists of controls, which are located on $\sigma$ and
switched on with delay $T-s$\,\,(the value $s$ is an action time).
\smallskip

\noindent$\bullet$\,\,\,An {\it inner space} of the system is
${\cal H}=L_2(\Omega)$. The waves $u^f(\,\cdot\,,t)$ are time dependent
elements of ${\cal H}$.
\smallskip

\noindent$\bullet$\,\,\,In the system $\alpha^T$, the input
$\mapsto$ state correspondence is realized by a {\it control
operator} $W^T: {\cal F}^T \to {\cal H},\,\,\,{\rm
Dom\,}W^T={\cal M}^T$
$$ W^T f\,:=\,u^f(\,\cdot\,,T)\,.$$ A specifics of the system
governed by the {\it scalar} wave equation (\ref{acoust 1}) is
that $W^T$ is a bounded operator. Therefore one can extend it from
${\cal M}^T$ onto ${\cal F}^T$ by continuity that we assume to be done.
\smallskip

\noindent$\bullet$\,\,\, The input $\mapsto$ output map is
represented by a {\it response operator} $R^T: {\cal F}^T \to
{\cal F}^T,\,\,{\rm Dom~}R^T={\cal M}^T,$
$$
R^Tf:=\frac{\partial u^f}{\partial \nu}\bigg|_{\Gamma \times
[0,T]}~,
$$
where $\nu=\nu(\gamma)$ is an outward normal at $\gamma \in
\Gamma$.

The following evident fact was already mentioned in Introduction.
\begin{proposition}\label{Prop 8}
If two Riemannian manifolds have the mutual boundary and are
isometric (the isometry being identity at the boundary), then their (acoustical) response operators coincide. In
particular, for the manifold $\Omega$ and its copy $\tilde \Omega$
one has $R^{2T}=\tilde R^{2T}$ for any $T>0$.
\end{proposition}
\medskip

\noindent$\bullet$\,\,\, A {\it connecting operator} $C^T: {\cal
F}^T \to {\cal F}^T$ is defined by
\begin{equation}\label{connecting operator}C^T\,:=\,(W^T)^* W^T.
\end{equation}
By the definition, we have
\begin{equation*} (C^T f, g)_{{\cal F}^T} =
(W^T f, W^T g)_{\cal H} = \left(u^f(\,\cdot\,,T),
u^g(\,\cdot\,,T)\right)_{\cal H}\,,
\end{equation*}
i.e., $C^T$ connects the Hilbert metrics of the outer and inner
spaces. A significant fact is that the connecting operator is
determined by the response operator of the system $\alpha^{2T}$
through an explicit formula
\begin{equation}\label{C T via R 2T}
 C^T= \frac{1}{2}~(S^T)^* R^{2T} J^{2T}
S^T\,,
\end{equation}
 where the map $S^T: {\cal F}^T \to {\cal F}^{2T}$ extends the
controls from $\Gamma \times [0,T]$ to $\Gamma \times [0,2T]$ as
odd functions (of time $t$) with respect to $t=T$; $J^{2T}: {\cal
F}^{2T}\to {\cal F}^{2T}$ is an
integration:\,\,$(J^{2T}f)(\cdot,t)=\int_0^t f(\cdot,s)\,ds $ (see
\cite{BIP97}, \cite{BIP'07}).

\subsubsection*{Controllability}
The set ${\cal U}^s_\sigma:=\{u^f(\,\cdot\,,s)\,|\,\,f \in
{\cal F}^T_\sigma\}$ is said to be {\it reachable} (from $\sigma$, at
the moment $t=s$).

The operator $\Delta$, which governs the evolution of the system
$\alpha^T$, does not depend on time. By this, a time delay of
controls implies the same delay of the waves. As a result, one has
\begin{equation*}{\cal U}^s_\sigma\,=\,W^T {\cal F}^{T,s}_\sigma\,, \qquad 0 \leqslant s
\leqslant T\,.\end{equation*}

Problem (\ref{acoust 1})--(\ref{acoust 3}) is hyperbolic and the
finiteness of domains of influence does hold for its solutions:
for the delayed controls one has \begin{equation}\label{supp uf}{\rm
supp\,}u^f(\,\cdot\,, T)\subset \overline {\Omega^s[\sigma]}\,,
\qquad f \in {\cal F}^{T,s}_\sigma\,.\end{equation} The latter means that in the
system $\alpha^T$ the waves propagate with the unit velocity. As a
result, the embedding ${\cal U}^s_\sigma \,\subset\,
{\cal H}\langle\Omega^s[\sigma]\rangle$ is valid. The character of this
embedding is of principal importance: it turns out to be {\it
dense}. The following result is based upon the fundamental
Holmgren--John--Tataru uniqueness theorem (see \cite{BIP97},
\cite{BIP'07} for detail).
\begin{proposition}\label{Prop 9}
For any $s>0$ and $\sigma \in {\cal R}(\Gamma)$, the relation $\overline
{{\cal U}^s_\sigma}\,=\,{\cal H}\langle\Omega^s[\sigma]\rangle$ is valid (the
closure in ${\cal H}$). In particular, for $s=T>{\rm diam\,}\Omega$ one
has $\overline {{\cal U}^T_\sigma}\,=\,{\cal H}$.
\end{proposition}
In control theory this property is referred to as a {\it local
approximate boundary controllability} of the system $\alpha^T$. It
shows that the reachable sets are rich enough: any function
supported in the neighborhood $\Omega^s[\sigma]$ can be
approximated (in ${\cal H}$-metric) by a wave $u^f(\,\cdot\,,T)$ by
means of the proper choice of the control $f \in
{\cal F}^{T,s}_\sigma$.

By $P^s_\sigma$ we denote the projection in ${\cal H}$ onto the
reachable subspace $\overline{{\cal U}^s_\sigma}$ and call it a {\it
wave projection}. Recall that $X^s_\sigma$ is the projection in
${\cal H}$ onto ${\cal H}\langle\Omega^s[\sigma]\rangle$, which cuts off functions
onto the neighborhood $\Omega^s[\sigma]$. As a consequence of the
Proposition \ref{Prop 9} we obtain \begin{equation}\label{P=X}
P^s_\sigma\,=\,X^s_\sigma \,, \qquad \quad s>0,\,\,\,\sigma \in
{\cal R}(\Gamma)\,. \end{equation}

\subsection{IP of acoustics}
\subsubsection*{Setup} A dynamical inverse problem (IP) for the
system (\ref{acoust 1})--(\ref{acoust 3}) is set up as follows:

\noindent{\it given for a fixed $T>{\rm diam\,}\Omega$ the
response operator $R^{2T}$, to recover the mani\-fold $\Omega$}.

\noindent A physical meaning of the condition $T>{\rm
diam\,}\Omega$ is that the waves $u^f$, which prospect the
manifold from the parts $\sigma$ of its boundary, need big enough
time to fill the whole $\Omega$: see (\ref{supp uf}) and
(\ref{Omega r[A]= Omega}).
\smallskip

As was clarified in Introduction, {\it to recover} $\Omega$ means
to construct (via given $R^{2T}$) a Riemannian manifold, which has
the same boundary $\Gamma$, and possesses the response operator,
which is equal to $R^{2T}$. Speaking in advance, it will be shown
that $R^{2T}$ determines the copy $\tilde \Omega$. Thus, $\tilde
\Omega$ provides the solution to the IP.

\subsubsection*{Model}
As an operator connecting two Hilbert spaces, the control operator
$W^T: {\cal F}^T \to {\cal H}$ can be represented in the form of a {\it
polar decomposition} \begin{equation*}W^T\,=\,\Phi^T |W^T| \,,\end{equation*} where
$$|W^T|\,:=\,\left[\left(W^T\right)^*
W^T\right]^{\frac{1}{2}}\overset{(\ref{connecting
operator})}=\left(C^T\right)^{\frac{1}{2}}$$ and $\Phi^T: |W^T|f
\mapsto W^T f$ is an isometry from ${{\rm Ran}\,}|W^T|\subset {\cal F}^T$
onto ${{\rm Ran}\,}W^T \subset {\cal H}$ (see, e.g., \cite{BSol}). In what
follows we assume that $\Phi^T$ is extended by continuity to an
isometry from $\overline{{{\rm Ran}\,}|W^T|}$ onto
$\overline{{{\rm Ran}\,}W^T}$.

Recall that ${\cal U}^s_\sigma:=W^T {\cal F}^{T,s}_\sigma$ are the reachable
sets of the system $\alpha^T$ and $P^s_\sigma$ is the projection
in ${\cal H}$ onto $\overline{{\cal U}^s_\sigma}$.

Let us say the (sub)space $\tilde {\cal H}:=\overline {{{\rm Ran}\,}|W^T|}
\subset {\cal F}^T$ to be a {\it model inner space}, $\tilde
{\cal U}^s_\sigma:=|W^T| {\cal F}^{T,s}_\sigma \subset \tilde {\cal H}$ a {\it
model reachable set}. By $\tilde P^s_\sigma$ we denote the
projection in $\tilde {\cal H}$ onto $\overline{\tilde {\cal U}^s_\sigma}$
and call it a {\it model wave projection}.

The model and original objects are related through the isometry
$\Phi^T$. In particular, the definitions imply $\Phi^T \tilde
P^s_\sigma=P^s_\sigma \Phi^T$.
\smallskip

Now let $T>{\rm diam\,}\Omega$, so that $\Omega^T[\sigma]=\Omega$
holds for any $\sigma$. By Proposition \ref{Prop 9}, one has
$\overline {{{\rm Ran}\,}W^T}={\cal H}$. By this, the isometry $\Phi^T$
turns out to be a unitary operator from $\tilde {\cal H}$ onto ${\cal H}$.
Its inverse $U:=(\Phi^T)^*$ maps ${\cal H}$ onto $\tilde {\cal H}$
isometrically and $U P^s_\sigma=\tilde P^s_\sigma U$ holds.

Let $\tilde X^s_\sigma:=U X^s_\sigma U^*$ be the image (in $\tilde
{\cal H}$) of the cutting off projection. The property (\ref{P=X})
implies \begin{equation}\label{tilde P= tilde X} \tilde P^s_\sigma\,=\,\tilde
X^s_\sigma \,, \qquad \quad s>0,\,\,\,\sigma \in {\cal R}(\Gamma)\,. \end{equation}

\subsubsection*{Solving IP}
It suffices to show that the operator $R^{2T}$ determines the copy
$\tilde \Omega$. The procedure is the following.
\begin{enumerate}
\item Find the connecting operator by (\ref{C T via R 2T}).
Determine the operator $|W^T|=\left(C^T\right)^{\frac{1}{2}}$ and
the subspace $\tilde {\cal H}=\overline{{\rm Ran\,}|W^T|}\subset
{\cal F}^T$.

\item Fix a $\sigma \in {\cal R}(\Gamma)$ and $s \in (0,T]$. In $\tilde {\cal H}$
recover the model reachable set $\tilde {\cal U}^s_\sigma=|W^T|
{\cal F}^{T,s}_\sigma \subset \tilde {\cal H}$ and determine the
corresponding projection $\tilde P^s_\sigma$. By (\ref{tilde P=
tilde X}), we get the projection $\tilde X^s_\sigma$. Thus, the
map (\ref{map sigma tilde X}) is at our disposal.

\item By Proposition \ref{Prop 6}, this map determines the copy
$\tilde \Omega$. Its response operator $\tilde R^{2T}$ coincides
with the given $R^{2T}$: see Proposition \ref{Prop 8}.
\end{enumerate}
The acoustical IP is solved.

\subsection{Maxwell system}
Here $\Omega$ is a smooth compact oriented Riemannian 3d-manifold.

Propagation of electromagnetic waves in a curved space is
described by the dynamical Maxwell system $\alpha^T_{\rm M}$
\begin{align}
\label{max 1}& e_{t}={\rm curl\,}h, \,\,h_{t}=-{\rm curl\,}e && {\rm in}\,\,\,
(\Omega\backslash \Gamma) \times (0,T)\\
\label{max 2}& e|_{t=0}=0,\,\,\, h|_{t=0}\,=\,0 && {\rm in}\,\,\,\,\Omega\\
\label{max 3}& e_\theta\,=\,f                  && {\rm
on}\,\,\,\Gamma \times [0,T]\,,
\end{align}
where $e_\theta:=e-e\cdot \nu\,\nu$ is a tangent component of $e$
at the boundary, $f$ is a time-dependent tangent field on
$\Gamma$\, ({\it boundary control}), $e$ and $h$ are the electric
and magnetic components of the solution. For controls of the
smooth class
$${\cal M}^T:=\left\{f \in \vec C^\infty\left(\Gamma
\times [0,T]\right)\,\big|\,\,\nu \cdot f =0,\,\,\, {\rm supp\,} f
\subset \Gamma \times (0,T]\right\}\,,$$ problem (\ref{max
1})-(\ref{max 3}) has a unique classical smooth solution
$\{e^f(x,t),h^f(x,t)\}$. Note that the condition on ${\rm
supp\,}f$ means that $f$ vanishes near $t=0$.

Since a divergence is an integral of motion of the Maxwell
system, one has \begin{equation*} {\rm div\,} e^f(\,\cdot\,,t)=0,\,\,\,{\rm div\,} h^f(\,\cdot\,,t)=0, \qquad t
\geqslant 0\,.
\end{equation*}

\subsubsection*{Attributes}
$\bullet$\,\,\,An {\it outer space} of the system $\alpha^T_{\rm
M}$ is the space $${\cal F}^T:= \left\{f \in \vec L_2\left(\Gamma
\times [0,T]\right)\,\big|\,\,\nu \cdot f =0\right\}.$$ The smooth
class ${\cal M}^T$ is dense in ${\cal F}^T$.

The outer space contains the subspaces
$${\cal F}^{T, s}_\sigma:= \left\{f \in {\cal F}^T\,\big|\,\,{\rm supp\,}f \subset
\sigma \times [T-s,T]\right\}, \qquad \sigma \in {\cal R}(\Gamma)$$
of controls, which are located on $\sigma$ and switched on with
delay
$T-s$\,\,(the value $s$ is an action time). 
\smallskip

\noindent$\bullet$\,\,\,An {\it inner space} of the system is the
space ${\cal C} \oplus {\cal C}$. By (\ref{max 1}), the solutions
$\{e^f(\,\cdot\,,t), h^f(\,\cdot\,,t)\}$ are time dependent
elements of this space. Also, we select its electric part ${\cal C} \oplus
 \{0\}\ni e^f(\,\cdot\,,t)$.
\smallskip

\noindent$\bullet$\,\,\,The input $\mapsto$ state correspondence
is realized by a {\it control operator} $W^T_{\rm M}: {\cal F}^T
\to {\cal C}\oplus {\cal C},\,\,\,{\rm Dom\,}W^T_{\rm M}={\cal M}^T$, $
W^T_{\rm M} f\,:=\,\{e^f(\,\cdot\,,T), h^f(\,\cdot\,,T)\}\,.$ Its
electric part is $W^T: {\cal F}^T \to {\cal C}$,
$$W^T: f \mapsto e^f(\,\cdot\,,T)\,.$$
In contrast to the acoustical (scalar) system, $W^T_{\rm M}$ and
$W^T$ are unbounded (but closable) operators.

A reason to select an electric part of the system $\alpha^T_{\rm
M}$ is that it is the electric component, which is controlled at
the boundary: see (\ref{max 3}). By this, $e^f$ and $h^f$ are not
quite independent. Moreover, for 
$T<\inf\{r>0\,|\,\,\Omega^r[\Gamma]=\Omega\}$ the operator $W^T$
is injective and, hence, $e^f(\,\cdot\,,T)$ determines
$h^f(\,\cdot\,,T)$ \cite{BIP'07}, \cite{BD_1}.

\smallskip

\noindent$\bullet$\,\,\, The input $\mapsto$ output map of the
system $\alpha^T_{\rm M}$ 
is represented by a {\it response operator} $R^T: {\cal F}^T \to
{\cal F}^T,\,\,{\rm Dom~}R^T={\cal M}^T,$
$$
R^Tf:=\nu \wedge h^f \big|_{\Gamma \times [0,T]}\,.
$$
The following fact is quite evident.
\begin{proposition}\label{Prop 10}
If two Riemannian manifolds have the mutual boundary and are
isometric (the isometry being identity at the boundary), then their Maxwell response operators coincide. In
particular, for the manifold $\Omega$ and its canonical copy
$\tilde \Omega$ one has $R^{2T}=\tilde R^{2T}$ for any $T>0$.
\end{proposition}
\medskip

\noindent$\bullet$\,\,\, An electric {\it connecting operator}
$C^T: {\cal F}^T \to {\cal F}^T$ is introduced via a {\it
connecting form} $c^T,\,\,{\rm Dom\,}c^T={\cal M}^T \times {\cal M}^T$,
\begin{equation*}
c^T[f,g]\,:=\,\left(e^f(\,\cdot\,,T),
e^g(\,\cdot\,,T)\right)_{\cal C}\,=\,\left(W^T f, W^T g\right)_{\cal C}\,.
\end{equation*}
It is a Hermitian nonnegative bilinear form. As such, it is
closable, the closure $\bar c^T$ being defined on ${\cal N}^T\times
{\cal N}^T$, where ${\cal N}^T$ is a lineal in ${\cal F}^T$, ${\cal N}^T \supset
{\cal M}^T$. The form $\bar c^T$ determines a unique self-adjoint
operator $C^T$ by the relation
\begin{equation*}\label{connecting operator electric}
 (C^T f, g)_{{\cal F}^T}\,=\,\bar c^T[f, g]\,, \qquad f \in {\rm Dom\,}C^T, \,g \in {\cal N}^T
\end{equation*}
(see, e.g., \cite{BSol}). In fact, to close $c^T$ is to close
$W^T$, and one has ${\cal N}^T={\rm Dom\,}\bar W^T={\rm Dom\,}(C^T)^{1
\over 2}$. Hence, the knowledge of $c^T$ enables one to extend
$W^T$ from ${\cal M}^T$ to ${\cal N}^T$. In what follows this extension
(closure) is assumed to be done and denoted by the same symbol
$W^T$. The images $W^T f$ for $f\in {\cal N}^T$ are regarded as the
generalized solutions $e^f(\,\cdot\,,T)$.

As a result, one has the relations
\begin{equation}\label{closed connecting form}
\bar c^T[f,g]=\left((C^T)^{1 \over 2}f, (C^T)^{1 \over 2}
g\right)_{{\cal F}^T}=\left(W^T f, W^T g\right)_{\cal C}\,, \qquad f,g \in
{\cal N}^T\,.
\end{equation}
A key fact is that the connecting form is determined by the
response operator of the system $\alpha^{2T}_{\rm M}$ through an
explicit formula
\begin{equation}\label{C T via R 2T electric}
 c^T[f,g]= \left(2^{-1}(S^T)^* R^{2T} J^{2T}
S^T f, g \right)_{{\cal F}^T}\,, \qquad f, g \in {\cal M}^T\,,
\end{equation}
where the map $S^T: {\cal F}^T \to {\cal F}^{2T}$ extends the
controls from $\Gamma \times [0,T]$ to $\Gamma \times [0,2T]$ as
odd functions (of time $t$) with respect to $t=T$; $J^{2T}: {\cal
F}^{2T}\to {\cal F}^{2T}$ is an
integration:\,\,$(J^{2T}f)(\cdot,t)=\int_0^t f(\cdot,s)\,ds $ (see
\cite{BIP'07}).
\smallskip

Resuming the aforesaid, we can claim that $R^{2T}$ determines the
operator $(C^T)^{1 \over 2}$ by the scheme
\begin{equation}\label{scheme R2T to C 1/2}
R^{2T}\overset{(\ref{C T via R 2T electric})}\Rightarrow c^T
\Rightarrow \bar c^T \Rightarrow C^T \Rightarrow (C^T)^{1 \over
2}\,.
\end{equation}

\subsubsection*{Controllability}
The set ${\cal E}^s_\sigma:=\{e^f(\,\cdot\,,s)\,|\,\,f \in {\cal F}^T_\sigma
\cap {\cal M}^T\}$ is said to be {\it reachable} (from $\sigma$, at
the moment $t=s$).

The operators ${\rm curl\,}$, which govern the evolution of the system
$\alpha^T_{\rm M}$, does not depend on time. By this, a time delay
of controls implies the same delay of the waves. As a
result, one can represent \begin{equation*}{\cal E}^s_\sigma\,=\,W^T \left[{\cal F}^{T,s}_\sigma \cap
{\cal M}^T\right]\,.\end{equation*}

The Maxwell system (\ref{max 1})--(\ref{max 3}) obeys the
finiteness of domains of influence principle: for the delayed
controls  one has \begin{equation}\label{supp uf electric}{\rm
supp\,}e^f(\,\cdot\,, T)\subset \overline {\Omega^s[\sigma]}\,,
\qquad f \in \left[{\cal F}^{T,s}_\sigma \cap {\cal M}^T\right]\,.\end{equation} The
latter means that electromagnetic waves propagate with the unit
velocity. As a consequence, the embedding ${\cal E}^s_\sigma
\,\subset\, {\cal C}\langle\Omega^s[\sigma]\rangle$ is valid. Moreover, this
embedding is {\it dense}. This fact is derived from a vectorial
version of the Holmgren--John--Tataru uniqueness theorem (see
\cite{BIP'07} for detail).
\begin{proposition}\label{Prop 11}
For any $s>0$ and $\sigma \in {\cal R}(\Gamma)$, the relation $\overline
{{\cal E}^s_\sigma}\,=\,{\cal C}\langle\Omega^s[\sigma]\rangle$ is valid (the
closure in ${\cal C}$). In particular, for $s=T>{\rm diam\,}\Omega$ one
has $\overline {{\cal E}^T_\sigma}\,=\,{\cal C}$.
\end{proposition}
This property is interpreted as a {\it local approximate boundary
controllability} of the electric subsystem of $\alpha^T_{\rm M}$.

By $E^s_\sigma$ we denote the projection in ${\cal C}$ onto the
reachable subspace $\overline{{\cal E}^s_\sigma}$ and call it a {\it
wave projection}. Recall that $Y^s_\sigma$ is the projection in
${\cal C}$ onto ${\cal C}\langle\Omega^s[\sigma]\rangle$. As a consequence of the
Proposition \ref{Prop 11} we obtain \begin{equation}\label{E=Y}
E^s_\sigma\,=\,Y^s_\sigma \,, \qquad \quad s>0,\,\,\,\sigma \in
{\cal R}(\Gamma)\,. \end{equation}

\subsection{IP of electrodynamics}
\subsubsection*{Setup} A dynamical inverse problem (IP) for the
system (\ref{max 1})--(\ref{max 3}) is set up as follows:

\noindent{\it given for a fixed $T>{\rm diam\,}\Omega$ the
response operator $R^{2T}$, to recover the mani\-fold $\Omega$}.

\noindent A physical meaning of the condition $T>{\rm
diam\,}\Omega$ is the same as in the acoustical case: the
electromagnetic waves need big enough time to prospect the whole
$\Omega$: see (\ref{supp uf electric}) and (\ref{Omega r[A]=
Omega}).
\smallskip

As before, {\it to recover} $\Omega$ means to construct (via given
$R^{2T}$) a Riemannian manifold, which has the same boundary
$\Gamma$, and possesses the response operator, which is equal to
$R^{2T}$. As well as in the scalar case, we will show that
$R^{2T}$ determines the copy $\tilde \Omega$. Thus, $\tilde
\Omega$ will provide the solution to the IP.

\subsubsection*{Model}
Representing the (closed) control operator $W^T: {\cal F}^T \to {\cal C}$ in
the polar decomposition form, one has 
$W^T\,=\,\Psi^T |W^T|$, 
where
$|W^T|\,:=\,\left[\left(W^T\right)^*W^T\right]^{\frac{1}{2}}$ and
$\Psi^T: |W^T|f \mapsto W^T f$ is an isometry from
${{\rm Ran}\,}|W^T|\subset {\cal F}^T$ onto ${{\rm Ran}\,}W^T \subset {\cal C}$
\cite{BSol}. In what follows $\Psi^T$ is assumed to be extended by
continuity to an isometry from $\overline{{{\rm Ran}\,}|W^T|}$ onto
$\overline{{{\rm Ran}\,}W^T}$. Also note that (\ref{closed connecting
form}) implies $|W^T|=(C^T)^{\frac{1}{2}}$.
\smallskip

Recall that ${\cal E}^s_\sigma:=W^T [{\cal F}^{T,s}_\sigma \cap {\cal M}^T]$ is
an electric reachable set and $E^s_\sigma$ is the (wave)
projection in ${\cal C}$ onto $\overline{{\cal E}^s_\sigma}$.

Let us say the (sub)space $\tilde {\cal C}:=\overline {{{\rm Ran}\,}|W^T|}
\subset {\cal F}^T$ to be a {\it model inner space}, $\tilde
{\cal E}^s_\sigma:=|W^T| \left[{\cal F}^{T,s}_\sigma \cap {\cal M}^T\right]
\subset \tilde {\cal C}$ the {\it model reachable sets}. By $\tilde
E^s_\sigma$ we denote the projection in $\tilde {\cal C}$ onto
$\overline{\tilde {\cal E}^s_\sigma}$ and call it a {\it model wave
projection}.

The model and original objects are related through the isometry
$\Psi^T$. In particular, the definitions imply $\Psi^T \tilde
E^s_\sigma=E^s_\sigma \Psi^T$.
\smallskip

Now, let $T>{\rm diam\,}\Omega$. By Proposition \ref{Prop 11}, one
has $\overline {{{\rm Ran}\,}W^T}={\cal C}$. Therefore the isometry $\Psi^T$
turns out to be a unitary operator from $\tilde {\cal C}$ onto ${\cal C}$.
Its inverse $U:=(\Psi^T)^*$ maps ${\cal C}$ onto $\tilde {\cal C}$
isometrically and $U E^s_\sigma=\tilde E^s_\sigma U$ holds.

Let $\tilde Y^s_\sigma:=U Y^s_\sigma U^*$. The property
(\ref{E=Y}) implies \begin{equation}\label{tilde E= tilde Y} \tilde
E^s_\sigma\,=\,\tilde Y^s_\sigma \,, \qquad \quad s>0,\,\,\,\sigma
\in {\cal R}(\Gamma)\,. \end{equation}

\subsubsection*{Solving IP}
Let us show that the operator $R^{2T}$ determines the copy $\tilde
\Omega$.
\begin{enumerate}
\item Find the connecting form $c^T$ by (\ref{C T via R 2T
electric}). Determine the model control operator
$|W^T|=\left(C^T\right)^{\frac{1}{2}}$ (see (\ref{scheme R2T to C
1/2})) and the model inner space $\tilde {\cal C}=\overline{{\rm
Ran\,}|W^T|}\subset {\cal F}^T$.

\item Fix a $\sigma \in {\cal R}(\Gamma)$ and $s \in (0,T)$. In $\tilde {\cal C}$
recover the model reachable set $\tilde {\cal E}^s_\sigma=|W^T|
\left[{\cal F}^{T,s}_\sigma \cap {\cal M}^T\right]\subset \tilde {\cal C}$ and
determine the corresponding projection $\tilde E^s_\sigma$. By
(\ref{tilde E= tilde Y}), we get the projection $\tilde
Y^s_\sigma$. Thus, the map (\ref{map sigma tilde Y}) is at our
disposal.

\item By Proposition \ref{Prop 7}, this map determines the copy
$\tilde \Omega$. Its Maxwell response operator $\tilde R^{2T}$
coincides with the given $R^{2T}$ (see Proposition \ref{Prop 10}).
\end{enumerate}
The IP of electrodynamics is solved.

\subsection{Comments}
$\bullet$\,\,\,In this paper, the condition $T> {\rm diam\,}
\Omega$ is imposed for the sake of simplicity. It provides the
embedding $\check \tau_\sigma C(\Omega) \subset C(\Omega)$, which
is convenient just by technical reasons. However, there is a {\it
time-optimal} setup of the reconstruction problem, which takes
into account a local character of dependence of the acoustical and
Maxwell response operators on a near-boundary part of the
manifold. Namely, by the finiteness of the domain of influence,
for an arbitrary fixed $T>0$ the operator $R^{2T}$ is determined
by the submanifold $\Omega^T[\Gamma]$ (does not depend on the part
$\Omega \backslash\Omega^T[\Gamma]$). Therefore, the natural setup
is: {\it given for a fixed $T>0$ the operator $R^{2T}$, to recover
$\Omega^T[\Gamma]$}. In such a stronger form the problem is solved
in \cite{BIP'07} and \cite{BD_2}.
\smallskip

\noindent$\bullet$\,\,\,In reconstruction via a spectral triple
$\{{\cal A}, {\cal H}, {\cal D}\}$   (see \cite{Connes}, \cite{RenVar}), the
algebra provides a topological space (that is $\widehat {\cal A}$),
whereas the operator ${\cal D}$ encodes a Riemannian metric on
$\widehat {\cal A}$. The metric is recovered (via ${\cal D}$) by means of
the {\it Connes distance formula}. In our scheme, the object
responsible for the metric is a selected family of generators of
the algebra (that is the eikonals).
\smallskip

\noindent$\bullet$\,\,\,Dealing with the reconstruction problem
for a graph, one can introduce the straightforward analog of the
eikonal algebra ${\mathfrak T}$. However, this algebra turns out to be
noncommutative. By this, we have to deal with its {\it Jacobson
spectrum} $\widehat{\mathfrak T}$, which is the topologized set of the
primitive ideals of ${\mathfrak T}$ \cite{Mur}. As the known examples show,
its structure is related with geometry of the graph but the
relation is of rather implicit character. This challenging problem
is open yet. An intriguing fact is that in some examples the space
$\widehat{\mathfrak T}$ is non-Hausdorff. It contains "clusters", which are
the groups of nonseparable points. Presumably, the clusters of
$\widehat{\mathfrak T}$ are related with interior vertices of the graph.
\smallskip

\section{Appendix}

Here we give proof of Lemmas~\ref{Lemma 1}, \ref{Eff}, \ref{hatpihom}. 

The standard operations on vector fields on the manifold $\nabla,\,{\rm div\,},\,{\rm curl\,}$
are understood in the generalized sense.
Here are standard formulas of vector analysis:
\begin{align}
    {\rm div}\,(\varphi u) &= \nabla \varphi \cdot u + \varphi\,{\rm div}\, u, \label{divfu} \\
    {\rm div}\,(u\wedge v) &= {\rm curl}\, u \cdot v - u \cdot {\rm curl}\, v, \label{divuv} \\
    {\rm curl}\,(\varphi u) &= \nabla \varphi \wedge u + \varphi\,{\rm curl}\, u. \label{rotfu}
\end{align}
In~\eref{divfu} and \eref{rotfu} a function $\varphi$ is
Lipschitz; a field $u$ is locally integrable and its
divergence is also locally integrable. In~\eref{divuv} we may
suppose that $u$ or $v$ is Lipschitz, and the other field is
locally integrable and has locally integrable ${\rm curl}\,$.

\subsection{Proof of Lemma~\ref{Lemma 1}}

Let the field $z\in\vec{\cal H}$ satisfy ${\rm curl}\, z\in \vec{\cal H}$.
Following~\cite{Leis}, we say that the field $z$ satisfies the condition
\begin{equation}
    z_\theta|_\Gamma = 0,    \label{zte}
\end{equation}
if for any field $v\in\vec{\cal H}$, such that
${\rm curl}\, v\in \vec{\cal H}$, we have
$$
    (z, {\rm curl}\, v)_\Omega = ({\rm curl}\, z, v)_\Omega.
$$
Here and further in this section $(\cdot, \cdot)_U$ and
$\|\cdot\|_U$ means the inner product and the norm in $L_2(U)$ or
$\vec L_2(U)$. It can be shown, that due to smoothness of the
boundary $\Gamma$ it suffices to check this
condition only for $v\in \vec C^\infty(\Omega)$.

Introduce the space
$$
    F := \{ u \in \vec{\cal H} : {\rm div}\, u \in L_2(\Omega),\, {\rm curl}\, u \in \vec{\cal H},\, u_\theta|_\Gamma = 0\}
$$
with the norm
$$
    \|u\|_F^2 := \|u\|_\Omega^2 + \|{\rm div}\, u\|_\Omega^2 + \|{\rm curl}\, u\|_\Omega^2.
$$

The following result
is valid for an $\Omega \subset {\mathbb R}^3$ (see \cite{Leis},
section~8.4) and can be easily generalized on a smooth manifold.
\begin{theorem}
    \label{FL2}
    The embedding of the space $F$ to $\vec{\cal H}$ is compact.
\end{theorem}
Actually, the stronger fact holds true: the space $F$ coincides
with vector Sobolev space  $\vec H^1(\Omega)$, which 
is compactly embedded to $\vec{\cal H}$. However,
Theorem~\ref{FL2} will suffice for our purposes.
Theorem~\ref{FL2} is used in spectral analysis of the Maxwell
operator on compact manifolds (see, e.g., \cite{DF}).
\smallskip

Let us outline the scheme of the proof of
Lemma~\ref{Lemma 1}. We obtain estimates for $L_2$-norms of ${\rm
curl}\,$ and divergence of the difference $\check\tau_\sigma u -
\varepsilon_\sigma u$ by $L_2$-norm of $u\in{\cal C}$
(inequalities \eref{rot_est}, \eref{div_est}), and establish the
boundary condition~\eref{zte} on $\Gamma$ for this difference.
This means that the field $\check\tau_\sigma u -
\varepsilon_\sigma u$ belongs to $F$ with the corresponding norm
estimate, which implies that the operator $\check\tau_\sigma -
\varepsilon_\sigma$ restricted to $\cal C$ is compact (by
compactness of the embedding $F \subset\vec{\cal H}$).

In what follows we consider $X^s_\sigma$ as the
projections in $\vec{\cal H}$, which cut off fields on
$\Omega^s[\sigma]$.

We will use the following relations, which are
valid for any $T>0$:
$$
    \int_{[0, T]} s \, d X^s_\sigma = T X_\sigma^T- \int_{[0, T]} X^s_\sigma \, ds,
$$
$$
    \int_{[0, T]} s \, d Y^s_\sigma = T Y_\sigma^T - \int_{[0, T]} Y^s_\sigma \, ds.
$$
Along with~\eref{spect repr tau} this implies that for $T
> {\rm diam}\,\Omega$ we have
\begin{equation}
     (\varepsilon_\sigma - \check\tau_\sigma)\, y = \left(\int_{[0, T]} (X^s_\sigma - Y^s_\sigma) \, ds\right) y, \quad y \in {\cal C}.
    \label{diffKT}
\end{equation}
To prove Lemma~\ref{Lemma 1} we need to establish a compactness of
the operator, which acts from $\cal C$ to $\vec{\cal H}$ by
$$
    K_\sigma := \int_0^T (X^\xi_\sigma - Y^\xi_\sigma) \, d\xi
$$
(this integral is the same for any $T > {\rm diam}\,\Omega$).
Define a family of operators acting from $\cal C$ to $\vec{\cal
H}$ by
$$
    K^s_\sigma := \int_0^s (X^\xi_\sigma - Y^\xi_\sigma) \, d\xi,\quad
    0 \leqslant s < \infty.
$$
One can easily check the following relation 
\begin{equation}
    \left(\int_0^s X^\xi_\sigma \, d\xi\,\,y\right)(x) = \max\{ s - \tau_\sigma (x), 0 \}\,y(x)\,, \qquad x \in \Omega.
    \label{tauint}
\end{equation}
\begin{lemma}
    \label{lemmaNz}
    Choose $\sigma\subset \Gamma$ and $s>0$.
    Let a field $\beta \in \vec{\cal H}\langle \Omega^s[\sigma]\rangle $
    be smooth in $\Omega^s[\sigma]$ (in particular, smooth on the boundary
    $\Omega^s[\sigma] \cap \Gamma$) and orthogonal to ${\cal C}\langle \Omega^s[\sigma]\rangle $.
    Then for any $z\in\vec{C}^\infty(\Omega)$ one has
    $$
        (\beta, K^s_\sigma\, {\rm curl}\, z)_{\Omega^s[\sigma]} = (\beta, \nabla \tau_\sigma \wedge z)_{\Omega^s[\sigma]}.
    $$
\end{lemma}
\begin{proof}
    Let $0<s'<s$.
    By the absolute continuity of Lebesgue integral we have
    \begin{equation}
        (\beta, K^{s'}_\sigma\, {\rm curl}\, z)_{\Omega^{s'}[\sigma]} \to
        (\beta, K^s_\sigma\, {\rm curl}\, z)_{\Omega^s[\sigma]},
        \quad s' \to s-0.
        \label{Kslim}
    \end{equation}
    As is evident, $\beta$ is orthogonal to ${\cal C}\langle \Omega^\xi[\sigma]\rangle $ for $\xi \leqslant s$; therefore
    \begin{align*}
        &(\beta, K^{s'}_\sigma\, {\rm curl}\, z)_{\Omega^{s'}[\sigma]} =
        \int_0^{s'} d\xi\, (\beta, (X^\xi_\sigma - Y^\xi_\sigma)\, {\rm curl}\, z)_{\Omega^\xi[\sigma]} = \\
        &\int_0^{s'} d\xi\, (\beta, X^\xi_\sigma\, {\rm curl}\, z)_{\Omega^\xi[\sigma]}
        \overset{\eref{tauint}}=
        (\beta, (s'-\tau_\sigma)\, {\rm curl}\, z)_{\Omega^{s'}[\sigma]} =\\
        &((s'-\tau_\sigma)\, \beta, {\rm curl}\,
        z)_{\Omega^{s'}[\sigma]}\,.
    \end{align*}
    Define a Lipschitz function $h$ in $\Omega$ as follows
    $$
        h(x) := \max \{s'-\tau_\sigma(x), 0\}
    $$
    We have
    \begin{equation}
        ((s'-\tau_\sigma)\, \beta, {\rm curl}\, z)_{\Omega^{s'}[\sigma]} =
        (h \beta, {\rm curl}\, z)_\Omega
        \label{hgrrot}
    \end{equation}
    (the field $h \beta$ is defined in $\Omega$ since $h$ vanishes outside of $\Omega^{s'}[\sigma] \subset \Omega^s[\sigma]$).
    The field $h\beta$ is Lipschitz, as function $h$ is Lipschitz,
    and the field $\beta$ is smooth in the neighborhood of ${\rm supp}\, h$,
    so we can apply a formula of integration by parts to the right hand side in \eref{hgrrot}.
        Orthogonality of $\beta$ to ${\cal C}\langle \Omega^s[\sigma]\rangle $ implies
    \begin{equation}
        {\rm curl}\,\beta \,|_{\Omega^{s}[\sigma]} = 0, \quad 
        \beta_\theta|_{\Omega^{s}[\sigma]\cap\,\Gamma} = 0.
        \label{rotte}
    \end{equation}
    Due to the second equality we have $(h\beta)_\theta|_\Gamma = 0$.
    So the integral over $\Gamma$ in integration by parts vanishes.
    Applying the first equality in~\eref{rotte} and formula~\eref{rotfu}, we obtain:
    \begin{align*}
        &(h \beta, {\rm curl}\, z)_\Omega =
        ({\rm curl}\, (h \beta), z)_\Omega =
        (\nabla h\wedge \beta, z)_\Omega =
        ((-\nabla \tau_\sigma)\wedge \beta, z)_{\Omega^{s'}[\sigma]} = \\
        &(\beta, \nabla \tau_\sigma \wedge z)_{\Omega^{s'}[\sigma]}.
    \end{align*}
The latter term tends to $(\beta, \nabla \tau_\sigma
\wedge z)_{\Omega^s[\sigma]}$ as $s' \to s$. Taking
into account ~\eref{Kslim}, we obtain the required equality.
\end{proof}

Note that Lemma~\ref{lemmaNz} holds true if $\Omega^s[\sigma] = \Omega$.
\begin{lemma}
        Let $\sigma\subset\Gamma$.
        For a field $z\in\vec{C}^\infty(\Omega)$ we have
    \begin{equation}
        \label{krotkrot}
        (K_\sigma\, {\rm curl}\, z, K_\sigma\, {\rm curl}\, z)_\Omega =
        2\,(K_\sigma\, {\rm curl}\, z, \nabla \tau_\sigma \wedge z)_\Omega.
    \end{equation}
\end{lemma}
\begin{proof}
    We have
    \begin{align}
        &
        (K_\sigma {\rm curl}\, z, K_\sigma{\rm curl}\, z)_\Omega =
        \int_0^T ds\, ((X^s_\sigma - Y^s_\sigma)\, {\rm curl}\, z,
        K_\sigma {\rm curl}\, z)_\Omega =\notag\\
        &
        \int_0^T ds\, \int_0^T d\xi\,
        ((X^s_\sigma - Y^s_\sigma)\, {\rm curl}\, z, (X^\xi_\sigma - Y^\xi_\sigma)
        \, {\rm curl}\, z)_\Omega=\notag\\
        &
        2 \int_0^T ds\, \int_0^s d\xi\,
        ((X^s_\sigma - Y^s_\sigma)\, {\rm curl}\, z, (X^\xi_\sigma - Y^\xi_\sigma)
        \, {\rm curl}\, z)_\Omega =\notag\\
        & 2 \int_0^T ds\,
        ((X^s_\sigma - Y^s_\sigma)\, {\rm curl}\, z, K^s_\sigma\, {\rm curl}\, z)_{\Omega^s[\sigma]}.
        \label{2Re}
    \end{align}
    As is clear, the field $\beta := (X^s_\sigma - Y^s_\sigma)\, {\rm curl}\, z$
    is orthogonal to ${\cal C}\langle \Omega^s[\sigma]\rangle $.
    Moreover, it is smooth in $\Omega^s[\sigma]$, since it is solenoidal and satisfies~\eref{rotte}.
    So we can apply Lemma~\ref{lemmaNz} to the integrand:
    $$
        ((X^s_\sigma - Y^s_\sigma)\,{\rm curl}\, z, K_\sigma^s\,{\rm curl}\, z)_{\Omega^s[\sigma]} =
        ((X^s_\sigma - Y^s_\sigma)\,{\rm curl}\, z, \nabla \tau_\sigma \wedge z)_{\Omega^s[\sigma]}.
    $$
    Substituting this to~\eref{2Re}, we obtain
    \begin{align*}
        &(K_\sigma{\rm curl}\, z, K_\sigma{\rm curl}\, z)_\Omega =
        2 \int_0^T ds\, ((X^s_\sigma - Y^s_\sigma)\,{\rm curl}\, z, \nabla\tau_\sigma \wedge z)_{\Omega^s[\sigma]} =
        \\
        &2\, (K_\sigma\,{\rm curl}\, z, \nabla\tau_\sigma \wedge z)_\Omega.
    \end{align*}
\end{proof}

Applying~\eref{krotkrot} to
$z\in\vec{C}^\infty(\Omega)$, we obtain
\begin{align*}
    &\|K_\sigma {\rm curl}\, z\|^2_{\Omega} =
    2\,(K_\sigma {\rm curl}\, z, \nabla\tau_\sigma \wedge
    z)_{\Omega}\leqslant
    C\,\|K_\sigma {\rm curl}\, z\|_{\Omega} \cdot \|z\|_{\Omega}.
\end{align*}
Therefore,
\begin{equation}
    \label{krotestimate}
    \|K_\sigma {\rm curl}\, z\|_{\Omega} \leqslant C\, \|z\|_{\Omega}.
\end{equation}

\begin{lemma}
\label{corollrotk}
For any field 
$u\in{\cal C}$ the relations
\begin{equation}
    \|{\rm curl}\, (K_\sigma u)\|_\Omega \leqslant C\, \|u\|_\Omega
    \label{rot_est}
\end{equation}
and
\begin{equation}
    (K_\sigma u)_\theta |_\Gamma = 0
    \label{te_vanish}
\end{equation}
are valid.
\end{lemma}
\begin{proof}
    Let $z\in \vec{C}^\infty(\Omega)$.
    Operator $K_\sigma$ is self-adjoint by~\eref{krotestimate} and we have
    \begin{align*}
        &|(K_\sigma u, {\rm curl}\, z)_{\Omega}| = |(u, K_\sigma {\rm curl}\, z)_{\Omega}|
        \leqslant
        \|u\|_{\Omega} \cdot \|K_\sigma {\rm curl}\, z\|_{\Omega} \leqslant \\
        &C \|u\|_{\Omega} \cdot \|z\|_{\Omega}.
    \end{align*}
    Since $z$ is arbitrary this estimate implies~\eref{rot_est}.
    Since $z$ is not necessarily compactly supported, the equality \eref{te_vanish} holds true.
\end{proof}

\begin{lemma}
Let $\sigma\subset\Gamma$. For any field $u\in {\cal C}$
we have
\begin{equation}
    \|{\rm div}\, (K_\sigma u)\|_\Omega \leqslant C\, \|u\|_{\Omega}.
    \label{div_est}
\end{equation}
\end{lemma}
\begin{proof}
    By the definition of $K_\sigma$, for large enough $T$ we have
    $$
        K_\sigma u = \left(\int_0^T X^s_\sigma \, ds\right) u - \left(\int_0^T E^s_\sigma \, ds\right) u.
    $$
    The second term belongs to ${\cal C}$ and thus is solenoidal 
    in $\Omega$. 
    By~\eref{tauint} the first term is equal to $(T-\tau_\sigma) \, u$.
    Then by formula~\eref{divfu} we have
    $$
        {\rm div}\,(K_\sigma u) = {\rm div}\,((T-\tau_\sigma) \, u) = -\nabla\tau_\sigma \wedge u.
    $$
    This completes the proof.
\end{proof}

\begin{proof}[Proof of Lemma~\ref{Lemma 1}]
Suppose $u\in {\cal C}$. It follows from the
estimates~\eref{rot_est}, \eref{div_est} and boundary
condition~\eref{te_vanish} that
$$
    \| K_\sigma u \|_F \leqslant \widetilde C \, \| u \|_\Omega.
$$
Then by compactness of the embedding $F
\subset\vec{\cal H}$ (Theorem~\ref{FL2}) we
conclude that $K_\sigma \in {\mathfrak K}({\cal C};
\vec{\cal H})$. In view of \eref{diffKT} this completes the proof.
\end{proof}

\subsection{Proof of Lemma~\ref{Eff}}
    At first we prove Lemma for $f \in C^\infty(\Omega)$.

    Choose a finite open cover $\{U_j\}$ of the support of $f$ such that every set of this cover is
    $C^\infty$-diffeomorphic to a ball in case $U_j \cap \Gamma = \emptyset$ or to a semi-ball
    $\{x\in {\mathbb R}^3 : |x| < 1,\,\, x^3 \geqslant 0\}$ otherwise.
    Choose a partition of unity $\zeta_j \in C_0^\infty(U_j)$ such that
    $$
        0 \leqslant \zeta_j \leqslant 1, \quad
        \sum_j \zeta_j \,\Big|_{{\rm supp}\, f} = 1.
    $$
    It is clear that
    $$
        \check f - Y[f] = \sum_j ( \check{\zeta_j f} - Y[\zeta_j f]),
    $$
    and the functions $\zeta_j f$ belong to $C^\infty_0(U_j)$.
    Thus, it is necessary to prove the Lemma for a function $f$ supported
    in some open set $U$ $C^\infty$-diffeomorphic to a ball or a semiball. 
    In this case, for any $y\in\cal C$ we have
    \begin{equation}
        (f y - Y[f]\, y)|_U = \nabla p_y, \quad
        p_y\in H^1(U), 
        \label{npy}
    \end{equation}
    and if the set $U$ intersects with $\Gamma$, then the following equality holds true
    $$
        p_y|_{U\cap\Gamma} = {\rm const}.
    $$
    This can be easily obtained with the help of the Helmholtz decomposition in $U$.

    The function $p_y$ in~\eref{npy} is uniquely determined up to additive constant,
    which can be chosen so that
    \begin{equation}
        p_y|_{U\cap\Gamma} = 0
        \label{pyUGa}
    \end{equation}
    if $U\cap\Gamma \ne \emptyset$, and
    $$
        \int_U p_y \, dx = 0
    $$
    otherwise.
    The Friedrichs and Poincar\'e inequalities imply
    that,
    in the both cases, there is a constant $C$ such that
    $$
        \|p_{y}\|_U \leqslant C \|\nabla p_{y}\|_U = \| f y - Y [f]\, y \|_U \leqslant C \|\check f - Y [f] \| \cdot \|y\|.
    $$
    Therefore, the mapping $y \mapsto p_y$ is continuous from $\cal C$ to $H^1(U)$.

    Now assume that a sequence $y_n$ weakly converges to zero in $\cal C$.
    Then the sequence $p_{y_n}$ weakly converges to zero in $H^1(U)$,
    and due to compactness of the embedding $H^1(U) \subset L_2(U)$
    this implies
    \begin{equation}
        \| p_{y_n} \|_U \to 0, \quad n \to \infty.
        \label{pnto0}
    \end{equation}
    Next, we have
    $$
        \|f y_n - Y[f]\, y_n\|^2_{\Omega} =
        (f y_n, f y_n - Y[f]\, y_n)_{\Omega} =
        (f y_n, \nabla p_{y_n})_{\Omega}.
    $$
    In the last equality we used \eref{npy} and the inclusion ${\rm supp}\, f \subset U$.
    Integrating by parts in this inner product, and
    applying formula~\eref{divfu} and equality ${\rm div}\, y_n = 0$, we arrive at
    $$
        (f y_n, \nabla p_{y_n})_{\Omega} =
        -\int_U \nabla f \cdot y_n\, p_{y_n} \, dx \leqslant
        M \|y_n\|_{\Omega} \cdot \|p_{y_n}\|_U
    $$
    ($M$ depends only on $f$).
    Integral over $\partial U$ vanishes since $f$ vanishes on $\partial U\setminus\Gamma$ and in the
    case $U \cap\Gamma \ne \emptyset$
    we have~\eref{pyUGa}.
    The right hand side of the latter inequality tends to zero because the norms of $y_n$ are bounded and \eref{pnto0} takes place.
    Then, with regard to the result of the previous calculation, we
    get the relation
    $$
        \|f y_n - Y[f]\, y_n\|_{\Omega} \to 0, \quad n \to \infty,
    $$
    which shows that the operator $\check f - Y[f]$ is compact.

    Now let us consider the case $f\in C(\Omega)$.
    The function $f$ can be approximated in $C(\Omega)$ by functions $f_n \in C^\infty(\Omega)$.
    Operators of multiplication by $f_n$ tend to the operator of multiplication by $f$ in the operator norm.
    Hence, the operator $\check f - Y[f]$ is compact as a limit of compact operators.

\subsection{Proof of Lemma~\ref{hatpihom}}
Here we prove the following properties:
\begin{align*}
    \dot\pi (\alpha f + \beta g) & = \alpha \dot\pi (f) + \beta \dot\pi (g),\\
    \dot\pi (fg) &= \dot\pi (f)\, \dot\pi (g),\\
    \| \dot\pi (f) \| &= \|f\|,
\end{align*}
where $f, g \in C(\Omega)$, $\alpha, \beta \in {\mathbb R}$. The
first and second relations follow from
Lemma~\ref{Eff}. For example, consider the second
one. We show that
\begin{equation}
    Y[f]\, Y[g] - Y[fg] \in {\mathfrak K}.
    \label{diff_comp}
\end{equation}
By Lemma~\ref{Eff} we have
$$
    Y[f]\, Y[g] = (f + K_1)\, Y[g] =
    f Y[g] + K = f (g + K_2) + K = fg + \widetilde K,
$$
where $K_1, K_2, K, \widetilde K \in {\mathfrak K}({\cal C},
\vec{\cal H})$. Applying Lemma~\ref{Eff} to the function $fg$, we
obtain~\eref{diff_comp}.

Consider the fourth property. We can restrict
ourselves with smooth $f$ since the mapping $\dot\pi$ is bounded.
The latter follows from the obvious inequality
$$
    \| \dot\pi(f) \| \leqslant \|f\|.
$$
Let us establish the opposite inequality.
We need to show that for any compact operator $K \in {\mathfrak K}$
we have
\begin{equation}
    \|Y[f] + K\| \geqslant \|f\|.
    \label{IfKf}
\end{equation}
Fix a point $x_0\in \Omega\setminus\Gamma$ such that $\nabla f(x_0) \ne 0$
(the case of a constant $f$ is trivial). 
Choose a sequence of functions $\varphi_j\in
C^\infty_0(\Omega\setminus\Gamma)$ such that ${\rm supp}\,
\varphi_j$ shrink to $x_0$ as $j\to\infty$.
Introduce the fields
$$
    y_j := \nabla f \wedge \nabla\varphi_j.
$$
Functions $\varphi_j$ can be chosen such that every field $y_j$
does not vanish identically. Owing to~\eref{divuv} we have ${\rm
div}\, y_j = 0$. Since ${\rm supp}\, y_j$ tend to $x_0$ as
$j\to\infty$, for sufficiently large $j$ the fields $y_j$ belong
to $\cal C$. Further, we have
$$
    f\, y_j = f \nabla f \wedge \nabla\varphi_j = \frac{1}{2} \nabla (f^2) \wedge \nabla\varphi_j,
$$
so by \eref{divuv} ${\rm div}\,(f y_j) = 0$ and for large $j$ the
fields $f y_j$ also belong to $\cal C$. Hence
\begin{equation}
    Y[f] y_j = Y (f y_j) = f y_j.
    \label{EUf}
\end{equation}
Consider a normed sequence
$$
    \tilde y_j = y_j / \|y_j\|.
$$
Obviously, the sequence $\tilde y_j$ weakly converges to zero in
$\cal C$. Therefore $K \tilde y_j \to 0$ in $\cal C$.
With regard to~\eref{EUf} this yields
$$
    \|(Y[f] + K)\, \tilde y_j\| = \|f \tilde y_j + K \tilde y_j\| \to |f(x_0)|, \quad j\to \infty.
$$
Since $\|\tilde y_j\| = 1$ we arrive at the inequality $\|Y[f] +
K\| \geqslant |f(x_0)|$. This occurs for all points
$x_0$, at which $f$ has nonzero gradient. So
\eref{IfKf} holds true.

\end{document}